\RequirePackage{fix-cm}
\documentclass[smallcondensed]{svjour3-HAL}     
\smartqed  

\usepackage{verbatim}
\usepackage{stackengine}
\usepackage{amssymb}
\usepackage{amsmath}
\usepackage{graphicx}
\usepackage{comment}
\usepackage[noend]{algorithmic}

\usepackage{amsmath}
\usepackage{wasysym}
\usepackage{booktabs}
\usepackage{todonotes}

\usepackage[ruled]{algorithm}
\usepackage{algorithmic}

\usepackage{url}
\urldef{\mailsa}\path|dominique.laurent@u-cergy.fr|
\urldef{\mailsb}\path|nicolas.spyratos@lri.fr|

\newcommand{\dashapprox}{\mid\hspace{-.07cm}\sim}
\begin{document}

\title{Handling Inconsistencies in Tables with Nulls and Functional Dependencies}

\titlerunning{Handling Inconsistencies in Tables with Nulls and Functional Dependencies}

%
%
\author{Dominique Laurent\and Nicolas Spyratos}

\institute{Dominique Laurent \at ETIS Laboratory - ENSEA, CY Cergy Paris University, CNRS\\F-95000 Cergy-Pontoise, France\\ \mailsa 
\and
Nicolas Spyratos \at LISN Laboratory - University Paris-Saclay, CNRS\\
F-91405 Orsay, France\\
\mailsb\\~\\
{\bf Acknowledgment:} Work conducted while the second author was visiting at FORTH Institute of Computer Science, Crete, Greece (https://www.ics.forth.gr/)}

%
\maketitle
\begin{abstract}
In this paper we address the problem of handling inconsistencies in tables with missing values (also called nulls) and functional dependencies. Although the traditional view is that table instances must respect all functional dependencies imposed on them, it is nevertheless relevant to develop theories about how to handle instances that violate some dependencies. Regarding missing values, we make no assumptions on their existence: a missing value exists only if it is inferred from the functional dependencies of the table. 

We propose a formal framework in which each tuple of a table is associated with a truth value among the following: true, false, inconsistent or unknown; and we show that our framework can be used to study important problems such as consistent query answering, table merging, and data quality measures - to mention just a few. In this paper, however, we focus mainly on consistent query answering, a problem that has received considerable attention during the last decades.

The main contributions of the paper are the following: (a) we introduce a new approach to handle inconsistencies in a table with nulls and functional dependencies, (b) we give algorithms for computing all true, inconsistent and false tuples, (c) we investigate the relationship between our approach and Four-valued logic in the context of data merging, and (d) we give a novel solution to the consistent query answering problem and compare our solution to that of table repairs. 

\begin{keywords}{Inconsistent database~.~Functional dependency~.~Null value~.~Data merging~.~Consistent query answering~}
\end{keywords}
\end{abstract}

\section{Introduction}\label{sec:intro}
In several applications today we encounter tables with missing values and functional dependencies.  Such a table is often the result of merging two or more other tables coming from different sources. Typical examples include recording the results of collaborative work, merging of tables during data staging in data warehouses or checking the consistency of a relational database. 

As an example of collaborative work consider two groups of researchers each studying three objects found in an archaeological site. The researchers of each group record in a table data regarding the following attributes of each object: 
\begin{itemize}
\item Identifier (here of the form $i_n$ where $n$ is an integer, distinct objects being associated with distinct identifiers)
\item Kind (such as statue, weapon, \ldots)
\item Material from which the object is made (such as iron, bronze, marble, \ldots)
\item Century in which the object is believed to have been made.
\end{itemize}
At the end of their work each group submits their findings to the site coordinator in the form of a table as shown in Figure~\ref{fig:ex-tables-intro} (tables $D_1$ and $D_2$). Each row of a table contains data recorded for a single object. For example, the row $(i_1, statue, marble, 1.BC)$ means that object $i_1$ is a statue made of marble and believed to have been made in the first century before Christ. Similarly the row $(i_2, statue,  , 2.BC)$ means that object $i_2$ is a statue of unknown material, believed to have been made during the second century before Christ. 
Note that, in this tuple, there is a missing value, meaning that the material from which object $i_2$ is made could not be determined. 

\begin{figure}[ht]
\begin{center}
{\footnotesize
\begin{tabular}{c|llll}
$D_1$&$Id$&$K$&$M$&$C$\\
\hline
&$i_1$&$k$&$m$&$c$\\
&$i_1$&&$m'$&\\
&$i_2$&$k'$&$m'$&$c$\\
&$i_2$&$k'$&$m''$&\\
&$i_3$&&$m$&\\
\end{tabular}
\qquad\qquad
\begin{tabular}{c|llll}
$D_2$&$Id$&$K$&$M$&$C$\\
\hline
&$i_1$&$k$&&$c$\\
&$i_2$&$k'$&&$c'$\\
&$i_2$&$k'$&$m''$&\\
&$i_3$&$k'$&\\
\end{tabular}

~\\~\\
\begin{tabular}{c|llll}
$D$&$Id$&$K$&$M$&$C$\\
\hline
&$i_1$&$k$&$m$&$c$\\
&$i_1$&&$m'$&\\
&$i_1$&$k$&&$c$\\
&$i_2$&$k'$&$m'$&$c$\\
&$i_2$&$k'$&$m''$&\\
&$i_2$&$k'$&&$c'$\\
&$i_3$&&$m$&\\
&$i_3$&$k'$&&\\
\end{tabular}
}
\end{center}
\caption{The tables prepared by the two groups and the merged table}
\label{fig:ex-tables-intro}
\end{figure}

Now, the data contained in the two tables can be merged into a single table $D$ containing all tuples from the two tables, without duplicates, as shown in Figure~\ref{fig:ex-tables-intro} (table $D$). In doing this merging, we may have discrepancies between tuples of $D$. For example, object $i_1$ appears in $D$ as being made from two different materials; and object $i_2$ appears as made from two different materials and in two different centuries. This kind of discrepancies may lead to `inconsistencies' that should be identified by the site coordinator and resolved in cooperation with the researchers of the two groups. 

It should be obvious from this example that the merging of two or more tables into a single table more often than not results in inconsistencies even if the individual tables are each consistent. For example, although each of the tables $D_1$ and $D_2$ shown in Figure~\ref{fig:ex-tables-intro} satisfies the functional dependencies $Id \to K$ and $Id \to C$, the merged table $D$ does not satisfy $Id \to C$. 

A similar situation arises in data warehouses where one tries to merge views of the underlying sources into a single materialized view to be stored in the data warehouse. 

As a last example, in a relational database, although each table may satisfy its functional dependencies, the database as a whole may violate some dependencies.
To determine whether the database is consistent with its dependencies, one proceeds as follows: all tables are merged by placing their tuples into a single universal table $D$ possibly with missing values (under certain assumptions discussed in \cite{Vardi88}); then all functional dependencies are applied on $D$ through the well known chase algorithm \cite{FaginMU82,Ullman}. If the algorithm terminates successfully ({\em i.e.,} no inconsistency is detected) then the database is consistent; otherwise the algorithm stops when a first inconsistency is detected and the database is declared inconsistent. 

So in general the question is: what should we do when a table is inconsistent? There are roughly three approaches: $(a)$ reject the table, $(b)$ try to correct or `repair' it so that to make it consistent (and therefore be able to work with the repaired table) and $(c)$ keep the table as is but make sure you know which part is consistent and which is not. 

The first approach is followed by database theorists when checking database consistency, as explained above. This approach is clearly not acceptable in practice as the universal table might contain a consistent set of tuples that can be useful to users ({\em e.g.,} users can still query the consistent part of the table).

The second approach tries to alleviate the impact of inconsistent data on the answers to a query by introducing the notion of repair: a repair is a minimally different consistent instance of the table and an answer is consistent if it is present in every repair. This approach, referred to as `consistent query answering', has motivated important research efforts during the past two decades and is still the subject of current research.  The reader is referred to Section~\ref{sec:query} for a brief overview of the related literature. However, this approach is always difficult to implement due to important issues related to computational complexity and/or to semantics (there is still no consensus regarding the definition of `consistent answer').

In our work we follow the third approach that is, we keep inconsistencies in the table but we determine which part of the table is consistent and which is not. More specifically, we use set theoretic semantics for tuples and functional dependencies that allow us to associate each tuple of the table with one truth value among the following: true, false, inconsistent or unknown. By doing so we can study a number of important problems including in particular the problem of consistent query answering, and the definition of data quality measures. 

Regarding consistent query answering, our model offers a fundamentally different and direct solution to the problem: the consistent answer is obtained by simply retrieving true tuples that fit the query requirements. 

Moreover our approach offers the possibility of defining meaningful data quality measures. For example if a table contains a hundred tuples of which only five are true while the remaining ones are inconsistent, then the quality of data contained in the table is five percent. Since we have polynomial algorithms for computing all true, false and inconsistent tuples, we can define several quality measures of the data contained in a table, inspired by the work in \cite{Parisi19}. We can then use such measures to accompany query answers so that users are informed of the quality of the answer they receive ({\em e.g.,} getting an answer from a table with ninety five per cent of true tuples is more reliable than if the table contained only five per cent of true tuples). 
However, defining and studying such measures lies outside the goals of the present paper. In this paper we focus on one important  application of our approach, namely consistent query answering. A complete account of data quality measures will be reported in a future paper. 

The main contributions of the present paper can be summarized as follows:
\begin{enumerate}
    \item 
We introduce a new approach to handle inconsistencies in a table with nulls and functional dependencies; we do so by adapting the set theoretic semantics of \cite{Spyratos87} to our context and by extending the chase algorithm so that all inconsistencies are accounted for in the table.
\item We give polynomial algorithms in the size of the table for computing all true and all inconsistent  tuples in the table.
\item We investigate the relationship of our approach with Four-valued logic in the context of data merging.
\item We propose a novel approach for consistent query answering and we investigate how our approach relates to existing approaches.
\end{enumerate}
The paper is organized as follows: In Section~\ref{sec:model} we recall basic definitions and notations regarding tables and we introduce the set theoretic semantics that we use in our work. In Section~\ref{sec:FD} we give definitions and properties regarding the truth values that we associate with tuples.  In Section~\ref{sec:chase} we study computational issues and give algorithms for computing the truth values of tuples. In Section~\ref{sec:four-logic}, we show how our approach relates to Four-value logic when merging two or more tables. In Section~\ref{sec:query} we present a novel solution to the problem of consistent query answering and compare it to existing approaches. Section~\ref{sec:conclusion} contains concluding remarks and suggestions for further research.

\section{The Model}\label{sec:model}
In this section we present the basic definitions regarding tuples and tables as well as the set theoretic semantics that we use for tuples and functional dependencies. Our approach builds upon earlier work on the partition model \cite{Spyratos87}.
\subsection{The Partition Model Revisited}
Following \cite{Spyratos87}, we consider a universe $U =\{A_1, \ldots , A_n\}$ in which every attribute $A_i$ is associated with a set of atomic values called the domain of $A_i$ and denoted by $dom(A_i)$. An element of $\bigcup_{A \in U}dom(A)$ is called a {\em domain constant} or a {\em constant}. We call {\em relation schema} (or simply {\em schema}) any nonempty subset of $U$ and we denote it by the concatenation of its elements; for example $\{A_1, A_2\}$ is simply denoted by $A_1A_2$. Similarly, the union of schemas $S_1$ and $S_2$ is denoted as $S_1S_2$ instead of $S_1 \cup S_2$.

We define a {\em tuple} $t$ to be a partial function from $U$ to $\bigcup_{A \in U} dom(A)$ such that, for every $A$ in $U$, if $t$ is defined over $A$ then $t(A)$ belongs to $dom(A)$. The domain of definition of $t$ is called the {\em schema} of $t$, denoted by $sch(t)$. We note that tuples in our approach satisfy the {\em First Normal Form} \cite{Ullman} in the sense that each tuple component is an atomic value from an attribute domain.

Regarding notation, we follow the usual convention that, whenever possible, lower-case characters denote domain constants and  upper-case characters denote the corresponding attributes. Following this convention the schema of a tuple $t=ab$ is $AB$ and more generally, we denote the schema of $t$ as $T$.

Assuming that the schema of a tuple $t$ is understood, $t$ is denoted by the concatenation of its values, that is: $t=a_{i_1} \ldots a_{i_k}$ means that for every $j=1, \ldots ,k$, $t(A_{i_j})=a_{i_j}$, $a_{i_j}$ is in $dom(A_{i_j})$, and $sch(t)=A_{i_1} \ldots  A_{i_k}$.

We assume that for any distinct attributes $A$ and $B$, we have either $dom(A)=dom(B)$ or $dom(A) \cap dom(B)=\emptyset$. However, this may lead to ambiguity when two attributes have the same domain. Ambiguity can be avoided by prefixing each value of an attribute domain with the attribute name. For example, if $dom(A)=dom(B)$ we can say `an $A$-value $a$' to mean that $a$ belongs to $dom(A)$, and `a $B$-value $a$' to mean that $a$ belongs to $dom(B)$. In order to keep the notation simple we shall omit prefixes whenever no ambiguity is possible.

Denoting by ${\cal T}$ the set of all tuples that can be built up given a universe $U$ and the corresponding attribute domains, a {\em table} $D$ is a {\em finite} sub-set of ${\cal T}$ where duplicates are {\em not allowed}.

\smallskip
Given a tuple $t$, for every $A$ in $sch(t)$, $t(A)$ is also denoted by $t.A$ and more generally, for every subset $S$ of $sch(t)$ the restriction of $t$ to $S$, also called {\em sub-tuple} of $t$, is denoted by $t.S$. In other words, if $S \subseteq sch(t)$, $t.S$ is the tuple such that $sch(t.S)=S$ and for every $A$ in $S$, $(t.S).A=t.A$.

 Moreover, $\sqsubseteq$ denotes the `sub-tuple' relation, defined over ${\cal T}$ as follows: for any tuples $t_1$ and $t_2$, $t_1 \sqsubseteq t_2$ holds if $t_1$ is a sub-tuple of $t_2$. It is thus important to keep in mind that whenever $t_1 \sqsubseteq t_2$ holds, it is understood that $sch(t_1) \subseteq sch(t_2)$ also holds.

The relation $\sqsubseteq$ is clearly a partial order over ${\cal T}$. Given a table $D$, the set of all sub-tuples of the tuples in $D$ is called the {\em lower closure} of $D$ and it is defined by: ${\sf LoCl}(D) = \{q \in {\cal T}~|~(\exists t \in D)(q \sqsubseteq t)\}$. We shall call a table {\em reduced} if it contains only maximal tuples ({\em i.e.,} if no tuple in the set is sub-tuple of some other tuple in the set). 

\smallskip
The notion of  {\em ${\cal T}$-mapping}, as defined below, generalizes that of interpretation  defined in \cite{Spyratos87}.
\begin{definition}\label{def:interpretation}
Let $U$ be a universe. A {\em ${\cal T}$-mapping} is a mapping $\mu$ defined from $\bigcup_{A \in U}dom(A)$ to $2^{\mathbb N}$. A {\em ${\cal T}$-mapping} $\mu$ can be extended to the set ${\cal T}$ as follows: for every $t=a_{i_1} \ldots a_{i_k}$ in ${\cal T}$,  $\mu(t)= \mu(a_{i_1}) \cap \ldots \cap \mu(a_{i_k})$.

A ${\cal T}$-mapping $\mu$ is an {\em interpretation} if $\mu$ satisfies the {\em partition constraint} stating that for every $A$ in $U$, and for all distinct $a$ and $a'$ in $dom(A)$, $\mu(a) \cap \mu(a') =\emptyset$.
\end{definition}
We emphasize that in \cite{Spyratos87} interpretations provide the basic tool for defining true tuples: a tuple $t$ is said to be true in an interpretation $\mu$ if $\mu(t)$ is nonempty. 

To see the intuition behind this definition consider a relational table $D$ over $U$ and suppose that each tuple is associated with a unique identifier, say an integer. Now, for every $A$ in $U$ and every $a$ in $dom(A)$, define $\mu(a)$ to be the set of all identifiers of the tuples in $D$ containing $a$. Then $\mu$ is an interpretation as it satisfies the partition constraint. Indeed, due to the fact that, for every attribute $A$ in $U$, a tuple $t$ can not have more than one $A$-value, it is then impossible that $\mu(a) \cap \mu(a')$ be nonempty for any distinct values $a$, $a'$ in $dom(A)$.

Incidentally, if for every $A$ in $U$ we denote by $dom^*(A)$ the set of all $A$-values such that $\mu(a) \ne \emptyset$, then the set $\{\mu(a)~|~a \in dom^*(A)\}$ is a {\em partition} of $\bigcup_{a \in dom^*(A)}\mu(a)$ (whence the name ``partition model"). The following example illustrates this important feature.
\begin{example}\label{ex:partition}
Considering $U=\{A,B,C\}$ and $D=\{ab, bc, ac, a'b', b'c', abc\}$, the tuples in $D$ can be respectively assigned the identifiers $1$, $2$, $3$, $4$, $5$ and $6$. In that case, we have $\mu(a)=\{1,3,6\}$, $\mu(a')=\{4\}$, $\mu(b)=\{1,2,6\}$, $\mu(b')=\{4,5\}$, $\mu(c)=\{2,3,6\}$, $\mu(c')=\{5\}$, and $\mu(\alpha) = \emptyset$ for any constant $\alpha$ different than $a$, $a'$, $b$, $b'$, $c$ and $c'$.

It is clear that the ${\cal T}$-mapping $\mu$ is an interpretation and, since $dom^*(A)$, $dom^*(B)$ and $dom^*(C)$ are respectively equal to $\{a,a'\}$, $\{b,b'\}$ and $\{c,c'\}$, it is easy to see that $\{\mu(\alpha)~|~\alpha \in dom^*(A)\}$ is a partition of $\{1,3,4,6\}$, $\{\mu(\beta)~|~\beta \in dom^*(B)\}$ is a partition of $\{1,2,4,5,6\}$, and $\{\mu(\gamma)~|~\gamma \in dom^*(C)\}$ is a partition of $\{2,3,5,6\}$.

Moreover, extending $\mu$ to non unary tuples yields the following regarding the tuples in $D$: $\mu(ab)=\{1,6\}$, $\mu(bc)=\{2,6\}$, $\mu(ac)=\{3,6\}$, $\mu(a'b')=\{4\}$, $\mu(b'c')=\{5\}$, and $\mu(abc)=\{6\}$.\hfill$\Box$
\end{example} 
Summarizing our discussion, when dealing with consistent tables in \cite{Spyratos87}, only interpretations are relevant. In the present work, we follow the same idea, but we also extend the work of \cite{Spyratos87} so that we can deal with inconsistencies. As we shall see, non satisfaction of the partition constraint in Definition~\ref{def:interpretation} is the key criterion to characterize inconsistent tuples.
\subsection{Functional Dependencies}
The notion of functional dependency in our approach is defined as in \cite{Spyratos87}.
\begin{definition}\label{def:fd}
Let $U$ be a universe. A {\em functional dependency} is an expression of the form $X \to Y$ where $X$ and $Y$ are nonempty sub-sets of $U$.

A ${\cal T}$-mapping $\mu$  {\em satisfies} $X \to Y$, denoted by $\mu \models X \to Y$, if  for all tuples $x$ and $y$, respectively over $X$ and $Y$, the following holds: if $\mu(x) \cap \mu(y) \ne \emptyset$ then  $\mu(x) \subseteq \mu(y)$.
\end{definition}
Based on Definition~\ref{def:fd}, for all $X$ and $Y$ such that $X \cap Y=\emptyset$, and for every ${\cal T}$-mapping $\mu$, the following holds:

\smallskip\centerline{$\mu\models X \to Y$ if and only if $\mu \models X \to A$ for every $A$ in $Y$.}

\smallskip\noindent
This is so because, for every $x$ and $y$ such that $\mu(x) \cap \mu(y) \ne \emptyset$, $\mu(x) \subseteq \mu(y)$ holds if and only if $\mu(x) \subseteq \mu(a)$ holds for every constant $a$ in $y$.

Therefore without loss of generality we can assume that all  functional dependencies are of the form $X \to A$ where $A$ is an attribute not in $X$. Under this assumption, we consider pairs $\Delta = (D, {\cal FD})$ where $D$ is a table over $U$ and ${\cal FD}$ a set of functional dependencies over $U$, and we say that a ${\cal T}$-mapping $\mu$ satisfies $\Delta$, denoted by $\mu \models \Delta$, if  $(i)$ for every $t$ in $D$, $\mu(t) \ne \emptyset$, and $(ii)$ $\mu$ satisfies every $X\to A$ in ${\cal FD}$.

\smallskip
To see how our notion of functional dependency relates to the standard one in relational databases \cite{Ullman},  recall first that a relation $r$ over universe $U$ satisfies $X \to A$ if for all tuples $t$ and $t'$ in $r$ such that $t.X=t'.X$, we have $t.A=t'.A$.

In our approach, let $\Delta=(D, {\cal FD})$ and consider two tuples $t$ and $t'$ in $D$ such that $XA$ is a subset of $sch(t)$ and of $sch(t')$ and let $t.X=t'.X=x$. Then  for every ${\cal T}$-mapping $\mu$ such that $\mu \models \Delta$, $\mu(t)$ and $\mu(t')$ are nonempty, implying that $\mu(x) \cap \mu(t.A)$ and $\mu(x) \cap \mu(t'.A)$ are also nonempty. By Definition~\ref{def:fd}, this implies that $\mu(x)$ is a sub-set of $\mu(t.A)$ and of  $\mu(t'.A)$. As a consequence, assuming that $t.A \ne t'.A$ ({\em i.e.,} that $X \to A$ is not satisfied in the sense of the relational model), means that $\mu(t.A) \cap \mu(t'.A)$ is nonempty, and therefore {\em $\mu$ can not be an interpretation.}

Therefore if we restrict ${\cal T}$-mappings to be interpretations then the notion of functional dependency satisfaction in our approach is {\em the same} as that of relational databases. As we shall see, this observation supports the notion of consistency for $\Delta$, to be given later (in Definition~\ref{def:db-consistency}).

\smallskip
Given $\Delta=(D, {\cal FD})$ and tuples $t$, $t'$, $t''$, the following notations are extensively used in the remainder of the paper.

\smallskip\noindent
$-$ $\Delta \vdash t$, denotes that if $\mu \models \Delta$ then $\mu(t) \ne \emptyset$.
\\
$-$ $\Delta \vdash (t \sqcap t')$, denotes that if $\mu \models \Delta$ then $\mu(t) \cap \mu(t') \ne \emptyset$.
\\
$-$ $\Delta \vdash (t \preceq t')$ denotes that if $\mu \models \Delta$ then $\mu(t) \subseteq \mu(t')$.
\\
$-$ $\Delta \vdash (t \preceq t' \sqcap t'')$ denotes that if $\mu \models \Delta$ then $\mu(t) \subseteq \mu(t') \cap \mu(t'')$.

\smallskip\noindent
Given $\Delta=(D, {\cal FD})$, we now build  a particular ${\cal T}$-mapping $\mu$ such  that $\mu \models \Delta$ as follows: Let $\left(\mu_i\right)_{i \geq 0}$ be the sequence defined by the steps below:
\begin{enumerate}
\item
Associate each tuple $t$ with an identifier, $id(t)$, called the {\em tuple identifier} of $t$ (this can be an integer that identifies $t$ uniquely).
\item
Let $\mu_0$ be the mapping defined for every domain constant $a$ by:\\ $\mu_0(a) =\{id(t)~|~t \in D \mbox{ and }a \sqsubseteq t\}$.
\item
While there exists $X \to A$ in ${\cal FD}$, $x$ over $X$ and $a$ in $dom(A)$ such that $\mu_i(xa) \ne \emptyset$ and $\mu_i(x) \not\subseteq \mu_i(a)$, define $\mu_{i+1}$ by: $\mu_{i+1}(a) = \mu_i(a) \cup \mu_i(x)$ and $\mu_{i+1}(\alpha) = \mu_i(\alpha)$ for any other constant $\alpha$.
\end{enumerate}
\begin{lemma}\label{lemma:least-model}
For every $\Delta = (D, {\cal FD})$, the sequence $\left(\mu_i\right)_{i \geq 0}$ has a unique limit $\mu^*$ such that  $\mu^* \models \Delta$.
Moreover:
\begin{enumerate}
\item
For all $a_1$ and $a_2$ in the {\em same attribute domain} $dom(A)$, if $\mu^*(a_1) \cap \mu^*(a_2) \ne \emptyset$ then there exist $X \to A$ in ${\cal FD}$ and $x$ over $X$ such that $\mu^*(x) \ne \emptyset$ and $\mu^*(x) \subseteq \mu^*(a_1) \cap \mu^*(a_2)$.
\item
For all $\alpha$ and $\beta$, $\Delta \vdash (\alpha \sqcap \beta)$ holds if and only if $\mu^*(\alpha) \cap \mu^*(\beta) \ne \emptyset$ holds.
\end{enumerate}
\end{lemma}
\begin{proof}
See Appendix~\ref{append:least-model}.
\hfill$\Box$
\end{proof}
Given $\Delta=(D, {\cal FD})$, Lemma~\ref{lemma:least-model} shows the following:
\begin{enumerate}
\item There always exists a ${\cal T}$-mapping $\mu$ such that $\mu \models \Delta$.
\item When two constants from the same domain have common identifiers with respect to $\mu^*$ then this is due to a functional dependency.
\item For every tuple $t$, $\Delta \vdash t$ if and only if $\mu^*(t) \ne \emptyset$.
\end{enumerate}
It is important to note that the ${\cal T}$-mapping  $\mu^*$ as defined in Lemma~\ref{lemma:least-model} is not necessarily an interpretation as the following example shows.
\begin{example}\label{ex:closure0}
Let $U=\{A,B,C\}$ and $\Delta =(D, {\cal FD})$ where $D=\{ab, bc, abc'\}$ and ${\cal FD}=\{B \to C\}$.
Associating $ab$, $bc$ and $abc'$ respectively with 1, 2 and 3, $\mu^*$ is obtained as follows:
\\
$\bullet$ First, we have $\mu_0(a)=\{1,3\}$, $\mu_0(b)=\{1,2,3\}$, $\mu_0(c)=\{2\}$ and $\mu_0(c')=\{3\}$ and $\mu_0(\alpha)=\emptyset$ for any other domain constant $\alpha$.
\\
$\bullet$ Then, considering $B \to C$, we have $\mu_1(a)=\{1,3\}$, $\mu_1(b)=\mu_1(c)=\mu_1(c')=\{1,2,3\}$ and $\mu_1(\alpha)=\emptyset$ for any other domain constant $\alpha$.

Hence, $\mu^*=\mu_1$ and we remark that $\mu^*(c)\cap \mu^*(c') \ne \emptyset$, thus that $\mu^*$ is not an interpretation. Nevertheless, as stated by Lemma~\ref{lemma:least-model}, it is easy to see that $\mu^* \models \Delta$.
\hfill$\Box$
\end{example}
We note here that the authors of \cite{SpyratosL87} use a construction similar to that of Lemma~\ref{lemma:least-model} to define a minimal model of $\Delta$, called `query model', assuming that $D$ is consistent with ${\cal FD}$.

Now, in order to characterize when $\Delta \vdash (t \preceq a)$ holds, we introduce the notion of  {\em closure of a tuple $t$ in $\Delta$} inspired by the well known relational notion of closure of a relation scheme with respect to a set of functional dependencies \cite{Ullman}.
\begin{definition}\label{def:closure}
Given a database $\Delta=(D, {\cal FD})$ and a tuple $t$, the {\em closure of  $t$ in $\Delta$} (or {\em closure of $t$} for short, when $\Delta$ is understood), denoted by $t^+$, is the set of all domain constants $a$ such that $\Delta \vdash (t \preceq a)$ holds.
\end{definition}
We notice that, based on Definition~\ref{def:closure}, for every constant $a$ occurring in a tuple $t$ ({\em i.e.,} if $a \sqsubseteq t$ holds) then $a$ is in $t^+$, because, in this case, $\mu(t) \subseteq \mu(a)$ holds for every ${\cal T}$-mapping $\mu$. However constants not occurring in $t$ may also appear in $t^+$ due to functional dependencies, as shown in the following example.
\begin{example}\label{ex:closure}
Continuing Example~\ref{ex:closure0} where $U=\{A, B, C\}$ and $\Delta =(D,{\cal FD})$ with $D=\{ab, bc, abc'\}$ and ${\cal FD}=\{B \to C\}$, we show that $c$ belongs to $(ab)^+$.

Indeed,  for every $\mu$ such that $\mu \models \Delta$, we have $\mu(ab) \subseteq \mu(b)$ (since $b \sqsubseteq ab$) and $\mu(b) \subseteq \mu(c)$ (due to $B \to C$ and the fact that $\mu(b) \cap \mu(c) \ne \emptyset$ must hold). Hence, by transitivity, $\mu(ab) \subseteq \mu(c)$ holds, implying that $\Delta \vdash (ab \preceq c)$ holds, which  by Definition~\ref{def:closure}, means that $c$ belongs to  $(ab)^+$. It should also be noticed that a similar argument shows that $c'$ also belongs to $(ab)^+$.
\hfill$\Box$
\end{example}
Clearly computing the closure directly from its definition is inefficient. Algorithm~\ref{algo:closure} gives a method for computing the closure, since the following lemma states that  this algorithm correctly computes the closure.

{\small
\algsetup{indent=1.5em}
\begin{algorithm}[h!t]
\caption{Closure of $t$ \label{algo:closure}}
\begin{algorithmic}[1]
\REQUIRE
$\Delta = (D, {\cal FD})$ and a tuple $t$.

\ENSURE The closure $t^+$ of $t$

\STATE{$\Delta_t := (D_t, {\cal FD})$ where $D_t = D \cup \{t\}$}\label{line:cl-db-assign}
\STATE{$t^+ := \{a~|~a \sqsubseteq t\}$}\label{line:init-cl}
\WHILE{$t^+$ changes}\label{line:loop-cl}
    \FORALL{$X \to A \in {\cal FD}$}
        \FORALL{$x$ such that for every $b$ in $x$, $b \in t^+$ and $\Delta_t \vdash xa$}\label{line:test}
            \STATE{$t^+:= t^+ \cup \{a\}$}
        \ENDFOR
     \ENDFOR
\ENDWHILE
\RETURN{$t^+$}
\end{algorithmic}
\end{algorithm}
}

\begin{lemma}\label{lemma:inclusion}
Let $\Delta = (D, {\cal FD})$ and $t$ a tuple. Then Algorithm~\ref{algo:closure} computes correctly the closure $t^+$ of $t$. 
\end{lemma}
\begin{proof}
See Appendix~\ref{append:lemma-inclusion}.
\hfill$\Box$
\end{proof}
We draw attention on the fact that the database involved in Algorithm~\ref{algo:closure} is not $\Delta$ but the database $\Delta_t$ that can be seen as $\Delta$ in which the tuple $t$ has been added.

It should however be noticed that in case $\Delta \vdash t$, this distinction is not necessary because in this case, for every tuple $q$, $\Delta \vdash q$ holds if and only if $\Delta_t \vdash q$ holds. This is a consequence of the fact that, as seen in Appendix~\ref{append:lemma-inclusion}, if $\Delta \vdash t$ then for every ${\cal T}$-mapping $\mu$, $\mu \models \Delta$ holds if and only if $\mu \models \Delta_t$.

On the other hand, the following example shows that when $\Delta \not\vdash t$, the introduction of $\Delta_t$ instead of $\Delta$ is necessary for correctly computing $t^+$.
\begin{example}\label{ex:closure1}
Let $U=\{A,B,C\}$ and $\Delta =(D, {\cal FD})$ where $D=\{ac, b\}$ and ${\cal FD}=\{A \to B, B \to C\}$.

It is easy to see that when numbering the tuples in $D$ by $1$ for $ac$ and $2$ for $b$, the ${\cal T}$-mapping $\mu^*$ for $\Delta$ is defined by: $\mu^*(a) =\mu^*(c)=\{1\}$, $\mu^*(b)=\{2\}$ and $\mu^*(\alpha)=\emptyset$ for any other constant $\alpha$.

For $t=ab$, we argue that $c$ is in  $t^+$, that is, for every $\mu$ such that $\mu \models \Delta$, $\mu(ab) \subseteq \mu(c)$ holds. Indeed, this trivially holds if $\mu(ab) =\emptyset$ (as is the case with $\mu^*$), and otherwise the following proof can be done:
\\
$\bullet$ As $\mu(ab) \ne \emptyset$, $A \to B$ implies that $\mu(a) \subseteq \mu(b)$.
\\
$\bullet$ As $\mu(ac) \ne \emptyset$, $\mu(a) \subseteq \mu(b)$ implies $\mu(bc) \ne \emptyset$. Thus, $\mu(b) \subseteq \mu(c)$, due to $B \to C$.
\\
$\bullet$ Therefore, $\mu(a) \subseteq \mu(c)$, and since $\mu(ab) \subseteq \mu(a)$, we have $\mu(ab) \subseteq \mu(c)$.

\smallskip
On the other hand, computing $(ab)^+$ using a modified version of Algorithm~\ref{algo:closure} where $\Delta_t$ is replaced by $\Delta$ would output $a$ and $b$ in the closure. It should also be noticed that computing $(ab)^+$ using Algorithm~\ref{algo:closure} is as follows: by the statement on line~\ref{line:init-cl}, $a$ and $b$ are inserted into the closure, and then, since for $t=ab$, $\Delta_t \vdash ab$ the above reasoning shows that $\Delta_t \vdash bc$ as well. Therefore, $c$ is inserted into the closure because the test line~\ref{line:test} succeeds.\hfill$\Box$
\end{example}
The following example shows a case where the tuple $t$ of which the closure is computed is such that $\Delta \vdash t$.
\begin{example}\label{ex:fd0}
As seen in Example~\ref{ex:closure0}, if $U=\{A,B,C\}$ and $\Delta =(D, {\cal FD})$ where $D=\{ab, bc, abc'\}$ and ${\cal FD}=\{B \to C\}$, $\mu^*$ is defined by: $\mu^*(a)=\{1,3\}$, $\mu^*(b)=\mu^*(c)=\mu^*(c')=\{1,2,3\}$ and $\mu^*(\alpha)=\emptyset$ for any other domain constant $\alpha$.

In this case, the computation of $(ab)^+$ according to Algorithm~\ref{algo:closure} is as follows:
\\
$\bullet$ As $ab$ is in $D$, $\Delta_t =\Delta$. We thus run Algorithm~\ref{algo:closure} with $\Delta$ instead of $\Delta_t$.
\\
$\bullet$ $(ab)^+$ is first set to $\{a, b\}$.
\\
$\bullet$ Considering $B \to C$, since $b$ is in $(ab)^+$, and since $\Delta \vdash bc$ and $\Delta \vdash bc'$ (this holds because $\mu^*(c)$ and $\mu^*(c')$ are nonempty), $c$ and $c'$ are inserted in $(ab)^+$.

As no further step is processed, $(ab)^+=\{a, b, c, c'\}$, as seen in Example~\ref{ex:closure}. Thus $\Delta \vdash (ab \preceq c)$ and $\Delta \vdash (ab \preceq c')$ hold, implying $\Delta \vdash (ab \preceq c \sqcap c')$.
\hfill$\Box$
\end{example}
\section{Semantics}\label{sec:FD}
In this section we provide basic definitions and properties regarding the truth value associated with a tuple. The following definition is borrowed from  \cite{Spyratos87}.
\begin{definition}\label{def:db-consistency}
$\Delta$ is said to be {\em consistent} if there exists an {\em interpretation} $\mu$ such that $\mu \models \Delta$.
\end{definition}
Since in our approach, inconsistent tables are {\em not} discarded, it is crucial to be able to provide semantics to any $\Delta = (D, {\cal FD})$, being it consistent or not. To this end, inspired by Belnap's Four-valued logic \cite{Belnap}, we consider {\em four} possible truth values for a given tuple $t$ in $\Delta$. The notations of truth values for tuples in our approach and their intuitive meaning are as follows, for a given tuple $t$:
\begin{itemize}
    \item Truth value ${\tt true}$: $t$ is true in $\Delta$.
    \item Truth value ${\tt false}$: $t$ is false in $\Delta$. This means that we do {\em not} follow the Closed World Assumption (CWA), according to which any non true tuple is false \cite{Reiter77}.
    \item Truth value ${\tt inc}$ ({\em i.e.,} inconsistent): $t$ is true {\em and} false in $\Delta$. This truth value is necessary for `safely' dealing with inconsistent tuples.
    \item Truth value ${\tt unkn}$ ({\em i.e.,} unknown): $t$ is not true, not false and not inconsistent  in $\Delta$. This truth value is necessary for dealing with tuples not falling in one of the previous three categories.
\end{itemize}
In order to formalize the exact meaning of these truth values in our approach, we introduce the following terminology and notation for a given tuple $t$:
\begin{itemize}
\item If $\Delta \vdash t$ holds, $t$ is said to be {\em potentially true} in $\Delta$. Notice here that by Lemma~\ref{lemma:least-model}, $t$ is potentially true if and only if $\mu^*(t) \ne \emptyset$. 
\item If $\Delta \vdash (t \preceq a \sqcap a')$ holds for some distinct $a$ and $a'$ in the same attribute domain, then we use the notation $\Delta \dashapprox t$, and in this case, $t$ is said to be {\em potentially false} to reflect that $\mu(a) \cap \mu(a')$ must be empty for $\mu$ to be an interpretation. By Definition~\ref{def:closure}, $\Delta \dashapprox t$ holds if and only if there exist $a$ and $a'$  in the same attribute domain such that $a$ and $a'$ are in $t^+$.
\end{itemize}
Consequently, if a tuple $t$ is such that $\Delta \vdash t$ and $\Delta \dashapprox t$, then for $\mu$ to be an interpretation, $\mu$ must associate $t$ with a set expected to be empty and nonempty, which is of course a case of inconsistency! This explains why, in our approach, `potentially true' and `potentially false', should respectively be understood as `true or inconsistent' and `false or inconsistent'.

Based on this intuition, each tuple is assigned one of the four truth values according to the following definition.
\begin{definition}\label{def:incons-tuple}
Given $\Delta = (D, {\cal FD})$ and a tuple $t$, the truth value of $t$ in $\Delta$, denoted by $v_\Delta(t)$, is defined as follows:

\smallskip\noindent
\begin{tabular}{ll}
$-$ $v_\Delta(t) ={\tt true}$&~~if $\Delta \vdash t$ and $\Delta \not\dashapprox t$; $t$ is said to be {\em true in $\Delta$}.\\
$-$ $v_\Delta(t) ={\tt false}$&~~if $\Delta \not\vdash t$ and $\Delta \dashapprox t$;
$t$ is said to be {\em false in $\Delta$}.\\
$-$ $v_\Delta(t) ={\tt inc}$&~~if $\Delta \vdash t$ and $\Delta \dashapprox t$;
$t$ is said to be  {\em inconsistent in $\Delta$}.\\
$-$ $v_\Delta(t) ={\tt unkn}$&~~if $\Delta \not\vdash t$ and $\Delta \not\dashapprox t$;
$t$ is said to be {\em unknown in $\Delta$}.
\end{tabular}
\end{definition}
We point out that the four truth values as defined above correspond exactly to the four truth values defined in the Four-valued logic \cite{Belnap}. The reader is referred to Section~\ref{sec:four-logic} for more details on this point. We illustrate Definition~\ref{def:incons-tuple} through the following example.
\begin{example}\label{ex:fd0-truth-value}
As in Example~\ref{ex:closure0}, let $U=\{A,B,C\}$ and $\Delta =(D, {\cal FD})$ where $D =\{ab, bc, abc'\}$ and ${\cal FD}=\{B \to C\}$.

It has been seen in Example~\ref{ex:fd0} that $(ab)^+=\{a, b, c, c'\}$. Thus $\Delta \dashapprox ab$ holds. Moreover, it is easy to see from Example~\ref{ex:closure0} that $\mu^*(ab)\ne \emptyset$, implying that $\Delta \vdash ab$ holds as well. As a consequence, by Definition~\ref{def:incons-tuple}, $v_\Delta(ab)={\tt inc}$, meaning that $ab$ is inconsistent in $\Delta$. We notice that similar arguments hold for $abc$, $abc'$, $bc$, $bc'$ and $b$, showing that these tuples are also inconsistent in $\Delta$.

Moreover, based on Definition~\ref{def:db-consistency}, we also argue that $\Delta$ is {\em not} consistent, because every $\mu$ such that $\mu \models \Delta$ cannot be an interpretation. This is so because $\mu^*(c) \cap \mu^*(c') \ne \emptyset$ and Lemma~\ref{lemma:least-model} imply that for  $\mu$ such that $\mu \models \Delta$, $\mu(c) \cap \mu(c')\ne \emptyset$.

\smallskip
Now, consider the tuple $bc''$ where $c''$ is a constant in $dom(C)$ distinct from $c$ and $c'$. To compute $(bc'')^+$ using Algorithm~\ref{algo:closure}, the database $\Delta_t = (D_t, {\cal FD})$ where $D_t=\{ab, bc, abc', bc''\}$ is first defined and then, the closure is first set to $\{b, c''\}$. The subsequent computation steps rely on $B \to C$ and on that $\Delta_t \vdash bc$ and $\Delta_t \vdash bc'$ to insert $c$ and $c'$ in the closure.

It therefore follows that $(bc'')^+=\{b, c'', c, c'\}$, thus that $\Delta \dashapprox bc''$ holds. Since $\Delta \not\vdash bc''$ (because $\mu^*(bc'')=\emptyset$ and $\mu^* \models \Delta$), it follows that $v_\Delta (bc'') ={\tt false}$. Hence $bc''$ and all its super-tuples are false in $\Delta$.

\smallskip
As an example of unknown tuple in $\Delta$, let $a'$ be in $dom(A)$ such $a' \ne a$, and consider $a'c$. Since $\mu^*(a'c)=\emptyset$, $\Delta \not\vdash a'c$. On the other hand, it can be seen that $(a'c)^+=\{a',c\}$, because $D \cup \{a'c\}$ does not allow any specific tuple derivation using $B \to C$. Hence, $\Delta \not\dashapprox a'c$, which shows that $v_\Delta (a'c) ={\tt unkn}$.\hfill$\Box$
\end{example}
The following example shows that computing {\em all} inconsistent tuples in $\Delta$ is not an easy task.
\begin{example}\label{ex:fds}
Let $\Delta =(D, {\cal FD})$ be defined over $U=\{A,B,C\}$ by $D=\{abc, ac'\}$ and ${\cal FD}=\{A\to B, B \to C\}$.

Here again, the tuples in $D$ along with the functional dependencies in ${\cal FD}$ show no explicit inconsistency. However computing $\mu^*$ yields the following:
\\
$\bullet$ To define $\mu_0$, we associate the tuples $abc$ and $ac'$ with the integers 1 and 2, respectively. It follows that $\mu_0(a) =\{1,2\}$, $\mu_0(b) =\{1\}$, $\mu_0(c) =\{1\}$, $\mu_0(c') =\{2\}$ and $\mu_0(\alpha)=\emptyset$ for any other domain constant $\alpha$.
\\
$\bullet$ The next steps modify $\mu_0$ so as to satisfy $A \to B$ and $B \to C$ as follows:
\begin{enumerate}
    \item Due to $A \to B$, $\mu_1$ is defined by: $\mu_1(a) =\{1,2\}$, $\mu_1(b) =\{1,2\}$, $\mu_1(c) =\{1\}$ and $\mu_1(c') =\{2\}$;
    \item Due to $B \to C$, $\mu_2$ is defined by: $\mu_2(a) =\{1,2\}$, $\mu_2(b) =\{1,2\}$, $\mu_2(c) =\{1, 2\}$ and $\mu_2(c') =\{1,2\}$.
\end{enumerate}
As $\mu_2 \models {\cal FD}$, $\mu^*=\mu_2$.  Moreover, we have $a^+=\{a, b, c, c'\}$ and $b^+=\{b,c,c'\}$ showing that, by Lemma~\ref{lemma:inclusion}, $\Delta \vdash a \preceq(c \sqcap c')$ and $\Delta \vdash (b \preceq  c \sqcap c')$, thus that $a$ and $b$ are inconsistent in $\Delta$. It can then be seen that, for example, $abc$, $bc'$ and $ac$ are also inconsistent in $\Delta$.

\smallskip
Now, let $\Delta_1=(D_1, {\cal FD})$ such that $D_1=\{ac, ac'\}$. In this case, $\mu^*$ is defined by $\mu^*(a)=\{1,2\}$, $\mu^*(c)=\{1\}$, $\mu^*(c')=\{2\}$ and $\mu^*(\alpha)=\emptyset$ for any other domain constant $\alpha$. Therefore, $a^+=\{a\}$, showing that $a$ is {\em not} inconsistent in $\Delta_1$. As a consequence, $ac$, $ac'$ along with all their sub-tuples are true in $\Delta_1$ and all other tuples are unknown in $\Delta_1$.
\hfill$\Box$
\end{example}
The following proposition shows that our notion of inconsistent tuple complies with Definition~\ref{def:db-consistency}.
\begin{proposition}\label{prop:consistent-db}
$\Delta=(D, {\cal FD})$ is consistent if and only if there exists no tuple $t$ such that $v_\Delta(t)={\tt inc}$.
\end{proposition}
\begin{proof}
We first note that if there exists a tuple $t$ such that $v_\Delta(t)={\tt inc}$, then $\Delta \vdash t$ and $\Delta \dashapprox t$. Hence  there exist $a$ and $a'$ in the same attribute domain $dom(A)$ such that $\Delta \vdash (t \preceq a\sqcap a')$. Thus every ${\cal T}$-mapping $\mu$ such that $\mu \models \Delta$ satisfies that $\mu(t) \ne \emptyset$ and $\mu(t) \subseteq \mu(a) \cap \mu(a')$, implying that $\mu(a) \cap \mu(a')\ne \emptyset$. Hence, $\mu$ is not an interpretation, showing that, by Definition~\ref{def:db-consistency},  $\Delta$ is not consistent. 

Conversely, assuming that there is no tuple $t$ such that $v_\Delta(t)={\tt inc}$, that is such that $\Delta \vdash t$ and $\Delta \dashapprox t$, we prove that $\mu^*$ is an interpretation of $\Delta$.  Indeed, if $a_1$ and $a_2$ are two constants in the same attribute domain $A$ such that $\mu^*(a_1) \cap \mu^*(a_2) \ne \emptyset$, then by Lemma~\ref{lemma:least-model}(1), there exist $X \to A$ in ${\cal FD}$ and $x$ over $X$ such that $\mu^*(x) \ne \emptyset$ and $\mu^*(x) \subseteq \mu^*(a_1) \cap \mu^*(a_2)$. Thus by Lemma~\ref{lemma:least-model}(2), for every $\mu$ such that $\mu \models \Delta$, $\mu(x) \subseteq \mu(a_1) \cap \mu(a_2)$ and $\mu(x) \ne \emptyset$. We therefore obtain that $\Delta \dashapprox x$ and $\Delta \vdash x$, thus that  $v_\Delta(x)={\tt inc}$, which is  a contradiction. Therefore $\mu^*$ is an interpretation, and the proof is complete.
\hfill$\Box$
\end{proof}
Based on Definition~\ref{def:incons-tuple}, we stress the following important remarks about potentially true and potentially false tuples in a given $\Delta =(D, {\cal FD})$:
\begin{itemize}
\item
Let $t$ be a potentially true tuple. Since $\Delta \vdash t$ holds, as a consequence of Lemma~\ref{lemma:least-model}, we have that $\mu^*(t) \ne \emptyset$. Therefore true or inconsistent tuples are those tuples that are associated with a nonempty set by {\em every} ${\cal T}$-mapping $\mu$ such that $\mu \models \Delta$. This implies that potentially true tuples in $\Delta$ are built up with constants occurring in $D$, and thus are in finite number. We provide in this paper effective algorithms for computing the sets of true tuples and  inconsistent tuples. 
\item
As potentially false tuples $t$ are such that $\Delta \dashapprox t$, they may not satisfy $\Delta \vdash t$. Hence, Lemma~\ref{lemma:least-model} cannot be used to characterize them. Moreover, if $\Delta \dashapprox t$, then  every tuple $t'$ such that  $t \sqsubseteq t'$ also satisfies $\Delta \dashapprox t'$. This is so because in this case, if $\Delta \vdash (t \preceq a \sqcap a')$, then $\Delta \vdash (t' \preceq a \sqcap a')$ holds as well. Thus, the number of potentially false tuples may be infinite in case some of the attribute domains are infinite.
\item
Moreover, since every false tuple  $t$ is potentially false and does not satisfy $\Delta \vdash t$, it also follows as above that every tuple $t'$ such that  $t \sqsubseteq t'$ is also false. Thus,  the number of false tuples may be infinite in case some of the attribute domains are infinite. However,  the following proposition allows to characterize when a given tuple $t$ is false.
\end{itemize}
\begin{proposition}\label{prop:false-tuple}
Given $\Delta =(D, {\cal FD})$ and a tuple $t$, $v_\Delta(t) = {\tt false}$ if and only if $\Delta \not\vdash t$ and $v_{\Delta_t}(t)= {\tt inc}$, where $\Delta_t=(D_t, {\cal FD})$ and $D_t=D\cup\{t\}$.
\end{proposition}
\begin{proof}
Assuming that $v_\Delta(t) = {\tt false}$ indeed entails that $\Delta \not\vdash t$ by Definition~\ref{def:incons-tuple}. Moreover, as $\Delta \dashapprox t$, there exist $A$ in $U$ and $a$ and $a'$ in $dom(A)$ such that $\Delta \vdash (t \preceq a \sqcap a')$. Now, given $\mu$ such that $\mu \models \Delta_t$,  it has been shown that  $\mu$ also satisfies that $\mu \models \Delta$ (see Appendix~\ref{append:lemma-inclusion}). Hence, $\mu(t) \subseteq \mu(a) \cap \mu(a')$ holds, which implies that $\Delta_t \vdash (t \preceq a \sqcap a')$, that is $\Delta_t \dashapprox t$. Since it holds that $\Delta_t \vdash t$, we obtain that $v_{\Delta_t}(t)= {\tt inc}$.

Conversely, if $v_{\Delta_t}(t)= {\tt inc}$ then, by Definition~\ref{def:incons-tuple}, $\Delta_t \vdash t$ and $\Delta_t \dashapprox t$ hold. Therefore, there exist $A$ in $U$ and $a$ and $a'$ in $dom(A)$ such that $\Delta_t \vdash (t \preceq a \sqcap a')$, which by Algorithm~\ref{algo:closure} and Lemma~\ref{lemma:inclusion}, implies that $a$ and $a'$ are in $t^+$. We thus obtain that  $\Delta \dashapprox t$, which combined with our hypothesis that $\Delta \not\vdash t$, implies that $v_\Delta(t) = {\tt false}$. The proof is therefore complete.
\hfill$\Box$
\end{proof}
\section{Computing the Semantics}\label{sec:chase}
Similarly to  standard two valued logic, where computing the semantics of $\Delta$ means computing the set of all tuples true in $\Delta$, in our approach, computing the semantics amounts to compute all true, inconsistent or false tuples, knowing that unknown tuples are the remaining ones. 

However, as mentioned above, the set of false tuples may be infinite, making it impossible to compute them all. In this work, the case of false tuples is only partially addressed, and we rather concentrate on potentially true tuples, with the goal of investigating  consistent query answering in our approach (see Section~\ref{sec:query}). 

\subsection{The Chase Procedure in our Approach}
We first propose an effective algorithm for the computation of all potentially true tuples in a given $\Delta$. This algorithm is in fact inspired by the standard chase algorithm \cite{Spyratos87,Ullman}, with the main difference that when a functional dependency cannot be satisfied, our algorithm does {\em not} stop.

{\small
\algsetup{indent=1.5em}
\begin{algorithm}[!t]
\caption{Chasing a table \label{algo:chase}}
\begin{algorithmic}[1]
\REQUIRE
$\Delta = (D, {\cal FD})$

\ENSURE The chased table $\Delta^*=(D^*,{\cal FD})$ and a set $inc({\cal FD})$ containing sets of tuples associated with each $X \to A$ in ${\cal FD}$

\STATE{$D^* := D$}\label{line:init-T}
\FORALL{$X \to A$ in ${\cal FD}$}
    \STATE{$inc(X \to A) := \emptyset$}
\ENDFOR
\WHILE{$D^*$ changes}\label{line:main-loop-chase}
    \FORALL{$X \to A \in {\cal FD}$}
        \FORALL{$t_1$ in $D^*$ such that $XA \subseteq sch(t_1)$}
            \FORALL{$t_2$ in $D^*$ such that $X \subseteq sch(t_2)$ and $t_1.X = t_2.X$}
               	 \IF{$A \not\in sch(t_2)$}
			\STATE{$D^* := D^* \cup \{t_2a_1\}$ where $a_1=t_1.A$ }\label{line:add-plus}
    	        \ENDIF
    	       	\IF{$A \in sch(t_2)$ and $t_1.A \ne  t_2.A$}\label{line:fd-violation}
			\STATE{Let $y_i= t_i.(sch(t_i)\setminus A)$ and $a_i=t_i.A$, for $i=1,2$}
			\STATE{$D^* := D^*  \cup \{y_2a_1\}$}\label{line:add-plus-bis}\\
			\COMMENT{$y_1a_2$ is inserted into $D^*$ when processing $t_2$ in place of $t_1$}\\ 
			\COMMENT{and $t_1$ in place of $t_2$}
    	          	\STATE{$inc(X\to A) := inc(X \to A) \cup \{x\}$ where $x = t_1.X = t_2.X$}\label{line:add-minus2}
    	       	\ENDIF
            \ENDFOR
        \ENDFOR
     \ENDFOR
\ENDWHILE
\COMMENT{Reduction: keep in $D^*$ only maximal tuples}
\FORALL{$t_1$ in $D^*$}\label{line:norm+}
    \FORALL{$t_2$ in $D^*$}
        \IF{$t_2 \sqsubseteq t_1$ and $t_1 \ne t_2$}
            \STATE{$D^*:= D^* \setminus \{t_2\}$}
        \ENDIF
    \ENDFOR
\ENDFOR
\STATE{$inc({\cal FD}) := \{inc(X \to A)~|~inc(X \to A) \ne \emptyset\}$}
\RETURN{$\Delta^*=(D^*,{\cal FD})$ and $ inc({\cal FD})$}
\end{algorithmic}
\end{algorithm}
}

Instead, our chasing algorithm carries on the computation, returning a database $\Delta^*=(D^*, {\cal FD})$ and a set $inc({\cal FD})$ based on which inconsistent and true tuples are shown to be efficiently computed. Before doing so, we illustrate Algorithm~\ref{algo:chase} in the context of our introductory example.

\begin{figure}[ht]
\begin{center}
{\small
\begin{tabular}{c|llll}
$D$&$Id$&$K$&$M$&$C$\\
\hline
&$i_1$&$k$&$m$&$c$\\
&$i_1$&&$m'$&\\
&$i_1$&$k$&&$c$\\
&$i_2$&$k'$&$m'$&$c$\\
&$i_2$&$k'$&$m''$&\\
&$i_2$&$k'$&&$c'$\\
&$i_3$&&$m$&\\
&$i_3$&$k'$&&\\
\end{tabular}
\qquad\qquad
\begin{tabular}{c|llll}
$D^*$&$Id$&$K$&$M$&$C$\\
\hline
&$i_1$&$k$&$m$&$c$\\
&$i_1$&$k$&$m'$&$c$\\
&$i_2$&$k'$&$m'$&$c$\\
&$i_2$&$k'$&$m''$&$c$\\
&$i_2$&$k'$&$m'$&$c'$\\
&$i_2$&$k'$&$m''$&$c'$\\
&$i_3$&$k'$&$m$&\\
\end{tabular}
}
\end{center}
\caption{The table $D$ and its chased version $D^*$}
\label{fig:ex-table-chased}
\end{figure}

\begin{example}\label{ex:runex-chase}
Running Algorithm~\ref{algo:chase} with the table $D$ shown in Figure~\ref{fig:ex-tables-intro}, and recalled in  Figure~\ref{fig:ex-table-chased}, produces the table $D^*$ shown in the right of Figure~\ref{fig:ex-table-chased} and the set $inc({\cal FD})=\{inc(Id \to M), inc(Id \to C)\}$ where $inc(Id \to M)=\emptyset$ and $inc(Id \to C) =\{i_2\}$. The main steps of the algorithm work as follows:
\\
$\bullet$
First, $D^*$ is assigned $D$, and  $inc(Id \to M)$ and $inc(Id \to C)$ are assigned $\emptyset$.
\\
$\bullet$
Due to the statement on line~\ref{line:add-plus}, the first two rows in $D$ (thus in $D^*$) generate the new tuples $(i_1,k,m')$ and $(i_1, m', c)$. Similarly, applying $Id \to K$ to the last two rows in $D$ generates the new tuple $(i_3,k',m)$.\\
The rows 4 and 5 in $D$ generate $(i_2, k', m'',c)$ and the rows 5 and 6 generate $(i_2, k', m'',c')$. Moreover, due to the statement on line~\ref{line:add-plus-bis}, the rows 4 and 6 generate $(i_2, k', m', c')$ and $(i_2, k',c)$ and  $i_2$ is inserted in $inc(Id \to C)$.
\\
$\bullet$
With these new tuples at hand, the loop on line~\ref{line:main-loop-chase} proceeds further, generating $(i_1,k,m',c)$ by the statement on line~\ref{line:add-plus}. No new tuple is generated at this stage.
\\
$\bullet$
The loop on line~\ref{line:main-loop-chase} is processed once again, producing no new tuple. When running the reduction step against the current state of $D^*$, the following tuples are removed: $(i_1, m')$, $(i_1,k,m')$, $(i_1, m', c)$, $(i_1,k,c)$, $(i_2,k',m'')$, $(i_2, k', c')$, $(i_2,k',c)$, $(i_3,m)$ and $(i_3,k')$.

Thus, the output of Algorithm~\ref{algo:chase} is indeed as expected. It is important to notice that, although tuples have been added in $D^*$ during the processing, the final number of tuples in $D^*$ is less than that in $D$. Although this particular result cannot be proven in general, it will be shown that in the worst case, the size of $D^*$ remains polynomial in the size of $D$.

We emphasize that some nulls present in $D$ have been replaced by actual values in $D^*$, thanks to the functional dependencies in ${\cal FD}$. For example the second tuple in $D$ with two nulls has been `completed' into a total tuple in $D^*$. However, such a completion has not been possible for {\em every} tuple in $D^*$. Namely, the $C$-value in the last tuple of $D^*$ is left as null.

Keeping in line with our statement that `a missing value exists only if it is inferred from the functional dependencies', this indicates that the $C$-value of this tuple could {\em not} be determined based on the content of $D$ and ${\cal FD}$, and no other conclusion can be drawn regarding this null.

To see why the two insertions mentioned in the statement on line~\ref{line:add-plus-bis} are needed, we first recall from \cite{Ullman} that, in the traditional case, the chased table characterizes the semantics of the input table, in case no inconsistency has been detected{\footnote{In traditional chase, the semantics of a table $D$ containing nulls is the set of all tuples true in every instance of $D$, {\em i.e.,} in every relation $R$ over $U$ with no nulls, that satisfies the functional dependencies and such that for every $t$ in $D$, there exists $r$ in $R$ such that $r.T=t$.}}. In this work our goal is similar, but has to be adapted to our context. Namely, we expect that the chased table $D^*$ can provide a syntactical characterization of all possibly true tuples in $\Delta$, that is of all tuples $t$ such that $\Delta \vdash t$ holds.

In the context of our example, if we assume that $(i_2, k',m',c')$ is not inserted during the processing then  Algorithm~\ref{algo:chase} would {\em not} fit our semantics. Indeed, for every ${\cal T}$-mapping $\mu$ such that $\mu \models \Delta$, $\mu(i_2) \subseteq \mu(c')$ holds because of $Id \to C$ applied to the seventh row in $D$. Thus, $\mu(i_2, k',m',c') =\mu(k',m',c')$, and since $\mu(k',m',c')$ is nonempty (due to the fifth row in $D$), $\Delta \vdash (i_2, k',m',c')$. Hence $(i_2, k',m',c')$ must appear in $D^*$ to fulfill our expectation.

Adding such `new' tuples when chasing a table is one of the main features of our approach, as compared with traditional chase. This step should be seen as a `by-product' of carrying on the computation even after encountering a violation of a functional dependency.
\hfill$\Box$ 
\end{example}
The following lemma shows that Algorithm~\ref{algo:chase} provides an operational means to characterize the tuples $t$ such that $\mu^*(t) \ne \emptyset$.
\begin{lemma}\label{lemma:chase}
Algorithm~\ref{algo:chase} applied to $\Delta=(D, {\cal FD})$ always terminates. Moreover, for every tuple $t$, $\mu^*(t) \ne \emptyset$ holds if and only if  $t$ is in ${\sf LoCl}(D^*)$.
\end{lemma}
\begin{proof}
See Appendix~\ref{append:lemma-chase}.\hfill$\Box$
\end{proof}
Recalling that ${\sf LoCl}(D^*)$ denotes the Lower Closure of $D^*$, that is the set of all sub-tuples of tuples in $D^*$, Lemma~\ref{lemma:chase} shows that $D^*$ is a `tabular' version of the set of all tuples $t$ such that $\mu^*(t) \ne \emptyset$, that is, by Lemma~\ref{lemma:least-model}, a `tabular' version of the set of all tuples $t$ such that $\Delta \vdash t$. Therefore, $D^*$ provides a syntactical characterization of the set of all tuples $t$ such that $\Delta \vdash t$, as expected in the previous example.
\subsection{Computing True Tuples and Inconsistent Tuples}
As mentioned just above, Lemma~\ref{lemma:least-model} and  Lemma~\ref{lemma:chase} show that, given $\Delta=(D, {\cal FD})$, a tuple $t$ is in ${\sf LoCl}(D^*)$ if and only if $\Delta \vdash t$ holds, that is, if and only if $t$ is potentially true in $\Delta$, that is if and only if $t$ is either true or inconsistent in $\Delta$.

\smallskip
To see how to compute the set of all inconsistent tuples, we first recall the notion of {\em closure of a relation scheme} as defined in relational database theory \cite{Ullman}.

Given a set ${\cal FD}$ of functional dependencies and a relation scheme $X$, the {\em closure of $X$ with respect to} ${\cal FD}$, or more simply the {\em closure of} $X$, denoted by $X^+$, is the set of all attributes $A$ in $U$ such that every table $D$ satisfying ${\cal FD}$ in the sense of relational tables, also satisfies $X \to A$.

\smallskip
It is well-known that $X^+$ is computed through the following two steps that are quite similar to the steps of Algorithm~\ref{algo:closure}:

{\small
\begin{itemize}
\item[]
$X^+:=X$
\item[]
{\bf while} $X^+$ changes {\bf do}

\hspace{.5cm} {\bf for all} $Y \to B$ in ${\cal FD}$ such that $Y \subseteq X^+$ {\bf do}

\hspace{1cm} $X^+:=X^+ \cup \{B\}$

{\bf return} $X^+$
\end{itemize}
}

\noindent
The following proposition shows a strong relationship between the closure of a relation scheme as recalled above and the closure of a tuple as stated in Definition~\ref{def:closure}.
\begin{proposition}\label{prop:closure}
Let $\Delta =(D, {\cal FD})$ and $t$ be such that $\Delta \vdash t$. For every tuple $q$ and every  $a$ in $dom(A)$ such that $q \sqsubseteq t$ and $a \sqsubseteq t$, we have: $a$ belongs to $q^+$ if and only if $A$ belongs to $Q^+$.
\end{proposition}
\begin{proof}
See Appendix~\ref{append:proof-prop-closure}.\hfill$\Box$
\end{proof}

{\small
\algsetup{indent=1.5em}
\begin{algorithm}[th]
\caption{Inconsistent tuples in  $\Delta=(D, {\cal FD})$ \label{algo:incons}}
\begin{algorithmic}[1]
\REQUIRE
The output of Algorithm~\ref{algo:chase}, that is $\Delta^* = (D^*, {\cal FD})$ and  $inc({\cal FD})$.

\ENSURE The set ${\sf Inc}(\Delta)$

\STATE{${\sf Inc}(\Delta) := \emptyset$}
\FORALL{$t$ in $D^*$}
    \FORALL{$X \to A$ in ${\cal FD}$ such that $XA \subseteq T$}
       \IF{$x= t.X$ is in $inc(X \to A)$}
            \FORALL{$q$ such that $q \sqsubseteq t$}\label{line:loop-inc}
            	\IF{$X \subseteq Q^+$}
			\STATE{${\sf Inc}(\Delta) := {\sf Inc}(\Delta) \cup \{t.Q\}$}\label{line:ins-inc}
                \ENDIF
            \ENDFOR
         \ENDIF
    \ENDFOR
\ENDFOR
\RETURN{${\sf Inc}(\Delta)$}
\end{algorithmic}
\end{algorithm}
}

\noindent
Using the notion of relation scheme closure, we introduce Algorithm~\ref{algo:incons} which computes the set of inconsistent tuples in $\Delta$. The correctness of this algorithm is shown in  Lemma~\ref{lemma:incons}.
\begin{lemma}\label{lemma:incons}
Given $\Delta=(D, {\cal FD})$, a tuple $t$ is inconsistent in $\Delta$ if and only if $t \in  {\sf Inc}(\Delta)$.
\end{lemma}
\begin{proof}
See Appendix~\ref{append:proof-prop-chase}.\hfill$\Box$
\end{proof}
The following proposition characterizes inconsistent and true tuples in $\Delta$ based on  Algorithm~\ref{algo:chase} and  Algorithm~\ref{algo:incons}.
\begin{proposition}\label{prop:chase}
Given $\Delta=(D, {\cal FD})$ and a tuple $t$:
\\
1. $t$ is inconsistent in $\Delta$ if and only if $t \in  {\sf Inc}(\Delta)$.
\\
2. $t$ is true in $\Delta$ if and only if $t \in {\sf LoCl}(D^*) \setminus {\sf Inc}(\Delta)$.
\end{proposition}
\begin{proof}
Immediate consequence of Definition~\ref{def:incons-tuple}, Lemma~\ref{lemma:chase} and Lemma~\ref{lemma:incons}. \hfill$\Box$
\end{proof}
The following examples illustrate Algorithm~\ref{algo:incons} and  Proposition~\ref{prop:chase}.
\begin{example}\label{ex:chase1}
As in Example~\ref{ex:fds}, let $\Delta =(D, {\cal FD})$ over $U=\{A,B,C\}$ where $D=\{abc, ac'\}$ and ${\cal FD}=\{A\to B, B \to C\}$. The tabular version of $D$ is shown on the left below, whereas $D^*$ is shown on the right.

\begin{center}
{\small
\begin{tabular}{c|lll}
$D$~~&$A$&$B$&$C$\\
\hline
&$a$&$b$&$c$\\
&$a$&&$c'$\\
\end{tabular}
\qquad\qquad\qquad\qquad
\begin{tabular}{c|lll}
$D^*$&$A$&$B$&$C$\\
\hline
&$a$&$b$&$c$\\
&$a$&$b$&$c'$\\
\end{tabular}
}
\end{center}
Running Algorithm~\ref{algo:chase}, $D^*$ is first set to $D$ and $abc'$ is inserted in $D^*$ by the statement line~\ref{line:add-plus} due to $A \to B$. Then, $b$ is inserted in $inc(B \to C)$ by the statement line~\ref{line:add-minus2}, due to the tuples $abc$ and $abc'$. Thus, the table $D^*$ output by Algorithm~\ref{algo:chase} is as shown above and $inc({\cal FD})= \{inc(A \to B), inc(B \to C)\}$ where $inc(A \to B)=\emptyset$ and $inc(B \to C)=\{b\}$.

When running Algorithm~\ref{algo:incons} for $abc$ in $D^*$, since $b$ is in $inc(B \to C)$, $b$, $ab$, $bc$ and $abc$ are inserted into ${\sf Inc}(\Delta)$, due to the statement on line~\ref{line:ins-inc}. This is so because the schema $Q$ of each of these tuples contains $B$, and so, satisfies $B \subseteq Q^+$.

Moreover, for $q=a$, due to $A \to B$, we have $A^+=ABC$ and thus, $B \subseteq A^+$ holds, showing that $a$ is inserted in  ${\sf Inc}(\Delta)$ on line~\ref{line:ins-inc}. A similar reasoning holds for $q=ac$ because $B \in (AC)^+$. Thus, $ac$ is also inserted in  ${\sf Inc}(\Delta)$ on line~\ref{line:ins-inc}. The only remaining possibility is $q=c$,  and does not modify ${\sf Inc}(\Delta)$ because $B \not\subseteq C^+$. A similar computation is performed with $abc'$ in $D^*$, adding $bc'$, $abc'$ and $ac'$ in ${\sf Inc}(\Delta)$. As no other tuple can be inserted in ${\sf Inc}(\Delta)$, Algorithm~\ref{algo:incons} returns

\smallskip
${\sf Inc}(\Delta)=\{abc,$ $abc',$ $ab$, $ac,$ $ac',$ $bc,$ $bc',$ $a,$ $b\}$, 

\smallskip\noindent
which, by Proposition~\ref{prop:chase}(1), is the set of all inconsistent tuples in $\Delta$. As a consequence,  by Proposition~\ref{prop:chase}(2), $c$ and $c'$ are the only true tuples in $\Delta$.

\smallskip
Now, as in Example~\ref{ex:fds}, referring to $\Delta_1 =(D_1 ,{\cal FD})$ with $D_1=\{ac, ac'\}$, it is easy to see that $D_1^*=D_1$. This implies that $\Delta_1$ is consistent, and that $ac$, $ac'$, $a$, $c$ and $c'$ are true in $\Delta_1$.
\hfill$\Box$
\end{example}
\subsection{The Case of False Tuples}
As already noticed, computing all tuples false in a given $\Delta=(D, {\cal FD})$ is not feasible in case of infinite attribute domains. However, given a tuple $t$ and assuming that $\Delta^*$ and ${\sf Inc}(\Delta)$ have been computed, Algorithm~\ref{algo:truth-value} allows to compute $v_\Delta(t)$. In this way, instead of being systematically identified, false tuples are identified on demand.

{\small
\algsetup{indent=1.5em}
\begin{algorithm}[t]
\caption{Tuple truth value in  $\Delta=(D, {\cal FD})$ \label{algo:truth-value}}
\begin{algorithmic}[1]
\REQUIRE
A tuple $t$, $\Delta^* = (D^*, {\cal FD})$ and  ${\sf Inc}(\Delta)$

\ENSURE The truth value of $t$ as one of the truth values {\tt true}, {\tt false}, {\tt inc} or {\tt unkn} 

\STATE{$v := {\tt unkn}$}
\IF{$t \in {\sf LoCl}(\Delta^*)$}\label{line:truth-value-test}
    \IF{$t \in {\sf Inc}(\Delta)$}
              \STATE{$v := {\tt inc}$}\label{line:i}
    \ELSE
        \STATE{$v := {\tt true}$}\label{line:t}
    \ENDIF
\ELSE
    \STATE{Compute $D^*_t$ using Algorithm~\ref{algo:chase} applied to $D^*\cup\{t\}$ and ${\cal FD}$}
    \STATE{Compute ${\sf Inc}(\Delta_t)$ using Algorithm~\ref{algo:incons} applied to $\Delta_t =(D^*\cup\{t\},{\cal FD})$}
    \IF{$t \in {\sf Inc}(\Delta_t)$}
    	  \STATE{$v := {\tt false}$}\label{line:f}
    \ENDIF
\ENDIF
\RETURN{$v$}
\end{algorithmic}
\end{algorithm}
}

\begin{proposition}\label{prop:truth-value}
Given $\Delta= (D, {\cal FD})$, and a tuple $t$ and assuming that $\Delta^*=(D^*, {\cal FD})$ and ${\sf Inc}(\Delta)$ have been computed,  the truth value returned by Algorithm~\ref{algo:truth-value} is equal to $v_\Delta(t)$.
\end{proposition}
\begin{proof}
Immediate consequence of Definition~\ref{def:incons-tuple} and Proposition~\ref{prop:false-tuple}. \hfill$\Box$
\end{proof}
The following example illustrates the algorithm.

\begin{example}\label{ex:truth-value}
It has been seen in Example~\ref{ex:runex-chase} that in the context of our introductory example, Algorithm~\ref{algo:chase} returns the table $D^*$  as shown in Figure~\ref{fig:ex-table-chased}, and  $inc(\Delta)=\{inc(Id \to K), inc(Id \to C)\}$ where $inc(Id \to K)=\emptyset$ and $inc(Id \to C)=\{i_2\}$. Thus, by Algorithm~\ref{algo:incons}, the set ${\sf Inc}(\Delta)$ is defined by:

\smallskip\begin{tabular}{rl}
${\sf Inc}(\Delta)=$&$\{t~|~i_2 \sqsubseteq t \sqsubseteq (i_2,k',m',c)\} \cup \{t~|~i_2 \sqsubseteq t \sqsubseteq (i_2,k',m',c')\}\, \cup$\\
&$\{t~|~i_2 \sqsubseteq t \sqsubseteq (i_2,k',m'',c)\} \cup \{t~|~i_2 \sqsubseteq t \sqsubseteq (i_2,k',m'',c')\}$\\
\end{tabular}

\smallskip\noindent
Applying now Algorithm~\ref{algo:truth-value}, we have the following:
\begin{itemize}
    \item $v_{\Delta}(i_1,a,m,c)=v_{\Delta}(i_1, k, m',c)={\tt true}$, because  line~\ref{line:t} changes the value of $v$, since these tuples are in $D^*$ but not in ${\sf Inc}(\Delta)$.
    \item $v_{\Delta}(i_2)=v_{\Delta}(i_2,c)=v_{\Delta}(i_2,c')={\tt inc}$, because  line~\ref{line:i} changes the value of $v$, since these tuples are in ${\sf Inc}(\Delta)$.
    \item $v_{\Delta}(i_1,k')=v_{\Delta}(i_1, c')={\tt false}$. Indeed, none of these tuples is in ${\sf LoCl}(D^*)$, thus implying that neither line~\ref{line:i} nor  line~\ref{line:t} applies. Moreover, when running Algorithm~\ref{algo:truth-value} with $(i_1,k')$ as input, $(i_1,k,m,c)$ and $(i_1,k')$ are in $D^*\cup\{(i_1,k')\}$. Hence $(i_1,k')$ is in ${\sf Inc}(\Delta_t)$, because of $Id \to K$, and by line~\ref{line:f}, $v_\Delta(k'm)$ is set to ${\tt false}$.
A similar reasoning holds for $(i_1, c')$, but using $Id \to C$. 
        \item $v_{\Delta}(k',m)={\tt unkn}$. Indeed, as above, when running Algorithm~\ref{algo:truth-value} with $(k',m)$ as input, lines~\ref{line:i} and \ref{line:t} do not change the value of $v$ (as $(k',m)$ does not occur in ${\sf LoCl}(D^*)$). Moreover, as  $D^* \cup \{(k',m)\}$ is consistent, the value of $v$ is not changed by the statement line~\ref{line:f}. Consequently, ${\tt unkn}$ is returned.\hfill$\Box$
\end{itemize}
\end{example}
\subsection{Complexity Issues}
We argue that the computation of inconsistent and true tuples in  $\Delta=(D, {\cal FD})$ is polynomial in the size of the table $D$ and in the order of the `number of conflicts with respect to functional dependencies' (to be defined shortly). To see this, denoting by $|E|$ the cardinality of a set $E$, we investigate the complexities of Algorithm~\ref{algo:chase} and of Algorithm~\ref{algo:incons}.

Regarding Algorithm~\ref{algo:chase}, we first notice that, contrary to the standard chase algorithm \cite{Ullman}, rows are added in the table during the computation, and some others are then removed by the reduction statement of line~\ref{line:norm+}. To assess the size of the table $D^*$ during the processing, we point out the following:
\begin{itemize}
    \item If no inconsistency is found during the processing of the while-loop on line~\ref{line:main-loop-chase}, at most one tuple is added in $D^*$ as the `join' of two tuples in $D$ by statement line~\ref{line:add-plus}. Therefore, the cardinality of $D^*$ remains in the same order as that of $D$. Notice in this respect that, upon reduction, {\em one} `join' tuple replaces {\em two} tuples in $D$, which reduces the size of the table $D^*$ output by the algorithm.
    \item However, when inconsistent tuples occur, the  statement line~\ref{line:add-plus}  adds more than one tuple and statement line~\ref{line:add-plus-bis} inserts tuples resulting from a cross-product.
\end{itemize}
To find an upper bound of the size of $D^*$, for every $X \to A$ in ${\cal FD}$, let $N(x)$ be the number of different $A$-values $a$ such that $\Delta \vdash xa$ and $x$ belongs to $inc(X \to A)$. We denote by $\delta$ the maximal value of all $N(x)$ for all $x$ in $inc({\cal FD})$; in other words $\delta= \max(\{N(x)~|~x \in inc({\cal FD})\})$. $\delta$ is precisely what was earlier referred to as the  `number of conflicts with respect to functional dependencies'.

Given a tuple in $D$ and a functional dependency $X \to A$ in ${\cal FD}$, each of the  statements line~\ref{line:add-plus} and line~\ref{line:add-plus-bis}  generates at most $\delta$ tuples. Since several functional dependencies may apply to $t$, at most $\delta^{|{\cal FD}|}$ tuples are generated for the given tuple $t$. Hence, the number of tuples generated by  the  statements lines~\ref{line:add-plus} and \ref{line:add-plus-bis}  is in ${\cal O}\left(|D|.\delta^{|{\cal FD}|}\right)$. We therefore obtain that the size of the table $D^*$ when running Algorithm~\ref{algo:chase} is in ${\cal O}\left(|D|.\left(1+ \delta^{|{\cal FD}|}\right)\right)$, that is in ${\cal O}\left(|D|.\delta^{|{\cal FD}|}\right)$.

\smallskip
Since the number of runs of the while-loop on line~\ref{line:main-loop-chase} is at most equal to the number of tuples added into $D^*$,  this number is in ${\cal O}\left(|D|.\delta^{|{\cal FD}|}\right)$. Since moreover one run of the while-loop is quadratic in the size of $D^*$,  the computational complexity  of this while-loop is in ${\cal O}\left(|D|^3.\delta^{3.|{\cal FD}|}\right)$.

The last point to be mentioned here is that the reduction processing on line~\ref{line:norm+} is performed through a scan $D^*$ whereby for every  $t$ in $D^*$ every sub-tuple of $t$ is removed. Such a processing being quadratic in the size of $D^*$, the overall computational complexity of Algorithm~\ref{algo:chase} is in ${\cal O}\left(|D|^3.\delta^{3.|{\cal FD}|}\right)$.

As the computational complexity of Algorithm~\ref{algo:incons} is clearly linear in the size of $D^*$, the global complexity of the computation of inconsistent and true tuples in $\Delta$ is as stated just above, and therefore {\em polynomial in the size of $D$}. 

\smallskip
Regarding Algorithm~\ref{algo:truth-value}, we notice that its complexity is in ${\cal O}\left(|D|^3.\delta^{3.|{\cal FD}|}\right)$ as well, because  it requires a scan of $D^*$ and then, in case the test line~\ref{line:truth-value-test} fails, Algorithm~\ref{algo:chase} is applied to a table whose cardinality is that of $D^*$ plus 1. It should however be kept in mind that, in this case, the algorithm has to be run {\em once for each tuple}, which shows that computing false tuples is not feasible even if all attribute domains are finite. Indeed, in this case, denoting by $DOM$ the maximal cardinality of attribute domains, the cardinality of  ${\cal T}$ is in ${\cal O}\left(|U|^{DOM}\right)$, thus yielding a computation in ${\cal O}\left(|U|^{DOM}.|D|^3.\delta^{3.|{\cal FD}|}\right)$.

We draw attention on the following important points regarding these complexity results:
\begin{enumerate}
\item Regarding the computation of false tuples, the above result has to be further investigated in the following two directions: first the computation of $D^*_t$ processed in Algorithm~\ref{algo:truth-value} is likely to be optimized using an {\em incremental} algorithm instead of Algorithm~\ref{algo:chase}, and second, it is expected that there exist interesting and relevant cases whereby the computation of $D^*_t$ is {\em not necessary}. We indeed suspect that this holds in  the case of a star schema. This is an important issue that lies out of the scope of the present paper, but that will be investigated in the next future.
    \item When the database is consistent, $\delta$ is equal to $1$, thus yielding a complexity in ${\cal O}(|D|^3)$. This result can be shown independently from the above computations as follows: In the case of traditional chase the maximum of nulls in $D$ being bounded by $|U|.|D]$, the number or iterations when running the algorithm is also bounded by $|U|.|D|$. Since the run of one iteration is in $|D|^2$, the overall complexity is in ${\cal O}(|U|.|D|^3)$, or in ${\cal O}(|D|^3)$, as $|U|$ is independent from $|D|$. 
    \item The above complexity study should be further investigated in order to provide more accurate results regarding the estimation of the number of actual tests necessary to the computation of $D^*$. The results in \cite{CKS86} are likely to be useful for such a more  thorough study of this complexity.
\end{enumerate}

\section{Four-Valued Logic and Table Merging}\label{sec:four-logic}
In this section, we first give a brief overview of  Belnap's Four-valued logic and then we show that our approach has a strong relationship with this formalism in the context of merging two or more tables.
\subsection{Basics of Four-Valued Logic}\label{subsec:four}
Four-valued logic was introduced by Belnap in \cite{Belnap}, who argued that his formalism is of interest when integrating data  from various data sources. To this end, he introduced four truth values denoted by {\tt t}, {\tt b}, {\tt n} and {\tt f} and read as {\em true}, {\em both true and false}, {\em neither true nor false} and {\em false}, respectively.  An important feature of this Four-valued logic is that its truth values can be compared according to two partial orderings, known as {\em truth ordering} and {\em knowledge ordering}, respectively denoted by $\preceq_t$ and $\preceq_k$ and defined as follows:

\begin{figure}[t]
\begin{center}
{\footnotesize
\begin{tabular}{c|c}
\,$\varphi$\,&\,$\neg \varphi$\,\\
\hline
${\tt t}$&${\tt f}$\\
${\tt b}$&${\tt b}$\\
${\tt n}$&${\tt n}$\\
${\tt f}$&${\tt t}$\\
\end{tabular}
\qquad
\begin{tabular}{c|llll}
\,$\vee$\,&\,{\tt t}&{\tt b}&{\tt n}&{\tt f}\\
\hline
{\tt t}&\,{\tt t}&{\tt t}&{\tt t}&{\tt t}\\
{\tt b}&\,{\tt t}&{\tt b}&{\tt t}&{\tt b}\\
{\tt n}&\,{\tt t}&{\tt t}&{\tt n}&{\tt n}\\
{\tt f}&\,{\tt t}&{\tt b}&{\tt n}&{\tt f}\\
\end{tabular}
\qquad
\begin{tabular}{c|llll}
\,$\wedge$\,&\,{\tt t}&{\tt b}&{\tt n}&{\tt f}\\
\hline
{\tt t}&\,{\tt t}&{\tt b}&{\tt n}&{\tt f}\\
{\tt b}&\,{\tt b}&{\tt b}&{\tt f}&{\tt f}\\
{\tt n}&\,{\tt n}&{\tt f}&{\tt n}&{\tt f}\\
{\tt f}&\,{\tt f}&{\tt f}&{\tt f}&{\tt f}\\
\end{tabular}
\\~\\~\\
\begin{tabular}{c|llll}
\,$\oplus$\,&\,{\tt t}&{\tt b}&{\tt n}&{\tt f}\\
\hline
{\tt t}&\,{\tt t}&{\tt b}&{\tt t}&{\tt b}\\
{\tt b}&\,{\tt b}&{\tt b}&{\tt b}&{\tt b}\\
{\tt n}&\,{\tt t}&{\tt b}&{\tt n}&{\tt f}\\
{\tt f}&\,{\tt b}&{\tt b}&{\tt f}&{\tt f}\\
\end{tabular}
\qquad
\begin{tabular}{c|llll}
\,$\otimes$\,&\,{\tt t}&{\tt b}&{\tt n}&{\tt f}\\
\hline
{\tt t}&\,{\tt t}&{\tt t}&{\tt n}&{\tt n}\\
{\tt b}&\,{\tt t}&{\tt b}&{\tt n}&{\tt f}\\
{\tt n}&\,{\tt n}&{\tt n}&{\tt n}&{\tt n}\\
{\tt f}&\,{\tt n}&{\tt f}&{\tt n}&{\tt f}\\
\end{tabular}
}
\end{center}
\caption{Truth tables of basic connectors}
\label{fig:truth-tables-con}
\end{figure}

\smallskip\centerline{
${\tt n}\preceq_k {\tt t}\preceq_k {\tt b}$~;~${\tt n}\preceq_k {\tt f} \preceq_k {\tt b}$ \qquad and \qquad 
${\tt f}\preceq_t {\tt n}\preceq_t {\tt t}$~;~${\tt f}\preceq_t {\tt b} \preceq_t {\tt t}$.}

\smallskip\noindent
As a consequence, two new connectors were introduced, denoted by $\oplus$ and $\otimes$, in addition to the standard connectors $\vee$ (disjunction) and $\wedge$ (conjunction). The corresponding truth tables, along with that for negation, are displayed in Figure~\ref{fig:truth-tables-con} and show that $\vee$ and $\oplus$ correspond to the least upper bound (lub) with respect to $\preceq_t$ and $\preceq_k$, respectively; whereas $\wedge$ and $\otimes$, correspond to the geatest lower bound (glb) with respect to $\preceq_t$ and $\preceq_k$, respectively . It is also shown in \cite{Belnap,Fitting91} that the set $\{{\tt t}, {\tt b}, {\tt n}, {\tt f}\}$ equipped with the two orderings $\preceq_t$ and $\preceq_k$ has a distributive bi-lattice structure.

Not surprisingly, some basic properties holding in standard logic do not hold in this setting. For example, Figure~\ref{fig:truth-tables-con} shows that formulas of the form $\Phi \vee \neg\Phi$ are not always true, independently of the truth value of $\Phi$. The reader is referred to the literature \cite{ArieliA98,Belnap,Fitting91,Laurent19,Tsoukias} for more details on the properties of Four-valued logic.

\smallskip
Based on the truth tables shown in Figure~\ref{fig:truth-tables-con}, it turns out that the connector $\oplus$ plays a key role in the context of data integration. Indeed, considering $n$ data sources $S_1, \ldots , S_n$ and a fact $\varphi$, for every $i=1, \ldots ,n$, $\varphi$ is assigned one truth value $v_i$, among ${\tt t}$,  ${\tt b}$, ${\tt n}$, or ${\tt f}$ in each $S_i$. The `integrated' truth value of $\varphi$, denoted by $v$ is then obtained as the expression $v=v_1\oplus \ldots \oplus v_n$, due to the following intuition:
\begin{itemize}
\item
The third row (or third column) of the truth table of $\oplus$ shows that every $v_i$ such that $v_i ={\tt n}$ plays no role in the resulting truth value $v$, provided that one of them be distinct from ${\tt n}$ (otherwise the `integrated' truth value of $\varphi$ is obviously ${\tt n}$). This fits our intuition that a source in which the truth value of $\varphi$ is unknown does not provide any piece of information regarding the `integrated' truth value of $\varphi$. We thus assume hereafter that for every $i=1, \ldots ,n$, $v_i \ne {\tt n}$.
\item
For every ${\tt v}$ among {\tt t}, {\tt b}, {\tt n} or {\tt f}, if $v_1= \ldots =v_n ={\tt v}$, then $v = {\tt v}$. The intuition here is that, since all sources agree on truth value ${\tt v}$, it is obvious to expect  $v$ to be this common value ${\tt v}$. For example, if for every $i=1, \ldots ,n$, $v_i ={\tt t}$, then it should be obvious that $v$ must be ${\tt t}$ as well!
\item
Now, if there exists $i_0$ such that $v_{i_0}={\tt b}$, then $v = {\tt b}$. This fits the intuition that if $\varphi$ is inconsistent in at least one data source, then  $\varphi$ remains inconsistent in the integrated source.
\item
The last case is when there exist distinct $i$ and $j$ in $\{1, \ldots ,n\}$ such that $v_i \ne v_j$, and no $v_i$ is equal to ${\tt b}$. In this case we have $v_i={\tt t}$ and $v_j={\tt f}$ (or equivalently  $v_i={\tt f}$ and $v_j={\tt t}$), which is the standard case of conflicting data sources in practice. In this case, it holds that $v={\tt b}$ (since ${\tt t} \oplus {\tt f}={\tt b}$). This result again fits our intuition that in case of conflicting data sources, the `integrated' truth value in inconsistent.
\end{itemize}
In the next sub-section, we show that, in our approach, the four truth values as defined in Definition~\ref{def:incons-tuple} also follow this intuition when it comes to merging two or more tables over the same universe $U$.
\subsection{Merging two or more Tables}\label{subsec:integration}
Data merging consists in collecting data from multiple, possibly heterogeneous sources and putting them in a single destination. The data from each source usually comes in the form of a CSV file, along with some hints on the data, referred to as metadata \cite{MEDES,RavatZ19}. During this process, different data sources are put together, or merged, into a single data store. Data merging is also related to data consolidation and to data integration.

When data comes from a broad range of sources, consolidation allows organizations to more easily present data, while also facilitating effective data analysis. Data consolidation techniques reduce inefficiencies, like data duplication, costs related to reliance on multiple databases and multiple data management points. 
 
In this section, we consider a simplified, relational scenario of $n$ sources $\Delta_1= (D_1, {\cal FD}_1), \ldots , \Delta_n= (D_n, {\cal FD}_n)$, where each source $\Delta_i= (D_i, {\cal FD}_i)$ consists of a table $D_i$ over a fixed universe $U$, possibly with nulls, and functional dependencies ${\cal FD}_i$. We then explain how to merge these sources in our approach under the following assumptions:

\begin{enumerate}
    \item All source tables are over the {\em same} universe $U$.
    \item Merging is done in the simplest possible way, namely $(a)$ the merged table is the {\em union} (in the set theoretic sense) of the source tables and $(b)$ the set of functional dependencies of the merged table is the union of the sets of functional dependencies  of the source tables. That is, the sources are merged through the pair: 
    $\Delta =(D, {\cal FD})$, where $D = \bigcup_{i=1}^{i=n}D_i$ and ${\cal FD}=\bigcup_{i=1}^{i=n}{\cal FD}_i$.
\end{enumerate}
Relying on Belnap's  Four-valued logic, we investigate the relationship between the truth values a tuple $t$ has in the source tables and the truth value the tuple $t$ has in the merged table.

First, notice that a `natural' one-to-one mapping $h$ from our set ${\sf Four}= \{{\tt true},$ ${\tt inc},$ ${\tt unkn},$ ${\tt false}\}$ to Belnap's set ${\cal FOUR}=\{{\tt t},$ $ {\tt b},$  ${\tt n},$ ${\tt f}\}$, can be defined by: $h({\tt true})={\tt t}$, $h({\tt inc})={\tt b}$, $h({\tt unkn})={\tt n}$ and $h({\tt false})={\tt f}$. Then, the connector $\oplus$ defined on ${\cal FOUR}$ induces a connector $\overline{\oplus}$ over ${\sf Four}$ defined by:
${\tt v}_1 \overline{\oplus}{\tt v}_2 = h^{-1}(h({\tt v}_1) \oplus h({\tt v}_2))$ for all ${\tt v}_1$ and ${\tt v}_2$ in ${\sf Four}$.

Moreover, we can define a partial ordering on ${\sf Four}$ isomorphic to the knowledge ordering of ${\cal FOUR}$ that allows us to compare truth values in ${\sf Four}$. Denoting this partial ordering by $\triangleleft$, we have:

\smallskip\centerline{${\tt unkn}\,\triangleleft\,{\tt false}\,\triangleleft\,{\tt inc}$\quad and\quad ${\tt unkn}\,\triangleleft\,{\tt true}\,\triangleleft\,{\tt inc}$}

\smallskip\noindent
The following proposition shows that the truth value of a tuple $t$ in the merged table is always greater (with respect to $\triangleleft$) than any of the truth values that $t$ has in the source tables in which it appears. In other words, when merging tables, the knowledge about tuples always increases, compared to the knowledge we have about tuples in the source tables.
\begin{proposition}\label{prop:integration}
Let $\Delta_i = (D_i,{\cal FD}_i)$ ($i=1, \ldots ,n$) be $n$ data sources over the same universe, and let $\Delta =(D,{\cal FD})$ be defined by $D = \bigcup_{i=1}^{i=n}D_i$ and ${\cal FD}=\bigcup_{i=1}^{i=n}{\cal FD}_i$. For every tuple $t$ the following holds:

\smallskip\centerline{
$ \overline{\bigoplus}\,_{i=1}^{i=n}\,v_{\Delta_i}(t)  ~\triangleleft~ v_\Delta(t)$.}
\end{proposition}
\begin{proof}
For every $i=1,\ldots ,n$, let $\Delta'_i=(D_i, {\cal FD})$. We first prove that for every tuple $t$, $v_{\Delta_i}(t)  ~\triangleleft~ v_{\Delta'_i}(t)$ holds. Indeed, for every $i=1,\ldots ,n$, let $D_i^*$, respectively $(D_i')^*$, the chased table of $D_i$ with respect to ${\cal FD}_i$, respectively ${\cal FD}$. Since ${\cal FD}_i \subseteq {\cal FD}$ holds, it is easy to see that for every $q_i$ in $(D_i')^*$ there exists $q$ in $D_i^*$ such that $q_i \sqsubseteq q$. Hence, for every $q$ in ${\cal T}$, $[q^+]_i \subseteq [q^+]'_i$, where $[q^+]_i$, respectively $[q^+]'_i$, denotes the closure of $q$ in $\Delta_i$, respectively $\Delta'_i$. Therefore, if $\Delta_i \vdash t$, respectively $\Delta_i \dashapprox t$, then $\Delta'_i \vdash t$, respectively $\Delta'_i \dashapprox t$, and so, for every $i=1, \ldots ,n$, $v_{\Delta_i}(t) \triangleleft v_{\Delta'_i} (t)$.

Considering $\Delta'_i$ ($i=1, \ldots ,n$) and $\Delta$, it can be seen that for every $i=1, \ldots ,n$ and every $q_i$ in $D_i'^*$ there exists $q$ in $D^*$ such that $q_i \sqsubseteq q$. Consequently, for every $i=1, \ldots ,n$, and every $q$ in ${\cal T}$, $[q^+]'_i \subseteq q^+$, where  $q^+$ denotes the closure of $q$ in $\Delta$. Therefore,  if for some $i$, $\Delta'_i \vdash t$, respectively $\Delta'_i \dashapprox t$, then $\Delta \vdash t$, respectively $\Delta \dashapprox t$, and so, for every $i=1, \ldots ,n$, $v_{\Delta'_i}(t) \triangleleft v_\Delta (t)$. The proposition follows from the transitivity of $\triangleleft$ and from the fact that $\overline{\oplus}$ defines the least upper bound (lub) with respect to $\triangleleft$, in the same way as $\oplus$ defines the lub with respect to $\preceq_k$.\hfill$\Box$
\end{proof}
In what follows, we identify cases where the equality $\overline{\bigoplus}\,_{i=1}^{i=n}\,v_{\Delta_i}(t)  = v_\Delta(t)$ holds and cases where it does not. To simplify, we assume that $n=2$. 

First, if for $i=1$ or $i=2$, $v_{\Delta_i}(t)={\tt inc}$, then the proposition implies that $v_\Delta(t)={\tt inc}$, because ${\tt inc}$ is maximal with respect to $\triangleleft$. In this case, the equality always holds. Another case where the equality holds is if $v_{\Delta_1}(t)={\tt true}$ and $v_{\Delta_2}(t)={\tt false}$. Indeed, in this case we have $\Delta \vdash t$ and $\Delta \dashapprox t$, showing that $v_\Delta(t)={\tt inc}$. Therefore, $v_\Delta(t)=v_{\Delta_1}(t)\,\overline{\oplus}\,v_{\Delta_2}(t)$.

To see cases where the equality $v_{\Delta_1}(t)\,\overline{\oplus}\,v_{\Delta_2}(t)  = v_\Delta(t)$ does not hold, let $U=\{A,B,C\}$, $\Delta_1=(\{abc\}, \emptyset)$ and $\Delta_2=(\{bc'\},\{B \to C\})$.

In this case, $\Delta =(D, {\cal FD})$ where $D= \{abc, bc'\}$ and ${\cal FD}=\{B \to C\}$. Hence, $D^*=\{abc,abc'\}$ and ${\sf Inc}(\Delta) =\{b, bc, bc', abc,abc'\}$, and so:
\begin{itemize}
\item $v_{\Delta_1}(b)=v_{\Delta_2}(b)={\tt true}$, whereas $v_\Delta(b)={\tt inc}$.
\item $v_{\Delta_1}(bc')= {\tt unkn}$, $v_{\Delta_2}(bc')={\tt true}$, thus implying that $v_1 \oplus v_2={\tt true}$, whereas $v_\Delta(bc')={\tt inc}$.
\end{itemize}
We further illustrate Proposition~\ref{prop:integration} in the the context of our introductory example.
\begin{example}\label{ex:truth-value-intro}
We recall that in our introductory example, we have two data sources $\Delta_1=(D_1, {\cal FD})$ and $\Delta_2=(D_2, {\cal FD})$, where ${\cal FD}=\{ID \to K, ID \to C\}$.

Based on $D_1$ and $D_2$ as shown in Figure~\ref{fig:ex-tables-intro} and displayed in Figure~\ref{fig:chase-tables-intro}, applying Algorithm~\ref{algo:chase} produces $D_1^*$ and $D_2^*$ also shown in Figure~\ref{fig:chase-tables-intro}, and returns ${\sf Inc}(\Delta_1)={\sf Inc}(\Delta_2)=\emptyset$.

\begin{figure}[ht]
\begin{center}
{\small
\begin{tabular}{c|llll}
$D_1$&$Id$&$K$&$M$&$C$\\
\hline
&$i_1$&$k$&$m$&$c$\\
&$i_1$&&$m'$&\\
&$i_2$&$k'$&$m'$&$c$\\
&$i_2$&$k'$&$m''$&\\
&$i_3$&&$m$&\\
\end{tabular}
\qquad\qquad
\begin{tabular}{c|llll}
$D_2$&$Id$&$K$&$M$&$C$\\
\hline
&$i_1$&$k$&&$c$\\
&$i_2$&$k'$&&$c'$\\
&$i_2$&$k'$&$m''$&\\
&$i_3$&$k'$&\\
\end{tabular}
~
\\~\\~\\~\\
\begin{tabular}{c|llll}
$D_1^*$&$Id$&$K$&$M$&$C$\\
\hline
&$i_1$&$k$&$m$&$c$\\
&$i_1$&$k$&$m'$&$c$\\
&$i_2$&$k'$&$m'$&$c$\\
&$i_2$&$k'$&$m''$&$c$\\
&$i_3$&&$m$&\\
\end{tabular}
\qquad\qquad
\begin{tabular}{c|llll}
$D_2^*$&$Id$&$K$&$M$&$C$\\
\hline
&$i_1$&$k$&&$c$\\
&$i_2$&$k'$&$m''$&$c'$\\
&$i_3$&$k'$&\\
\end{tabular}
}
\end{center}
\caption{The source tables of our introductory example and their chased versions}
\label{fig:chase-tables-intro}
\end{figure}

Hence, as already mentioned, $\Delta_1$ and $\Delta_2$ are consistent. Referring to Example~\ref{ex:truth-value} and Figure~\ref{fig:ex-table-chased}, applying Proposition~\ref{prop:integration} entails the following:
\begin{itemize}
    \item $v_{\Delta_1}(i_1,k,m,c)={\tt true}$, $v_{\Delta_2}(i_1,k,m,c)={\tt unkn}$ and $v_{\Delta}(i_1,k,m,c)={\tt true}$.\\
    $v_{\Delta_1}(i_1,k,m',c)={\tt true}$, $v_{\Delta_2}(i_1,k,m',c)={\tt unkn}$ and $v_{\Delta}(i_1,k,m',c)={\tt true}$.\\
    These are cases of equality because ${\tt true}\, \overline{\oplus}\, {\tt unkn}={\tt true}$.
    \item $v_{\Delta_1}(i_2, c)={\tt true}$, $v_{\Delta_2}(i_2,c) ={\tt false}$ and $v_{\Delta}(i_2,c) ={\tt inc}$.\\
    This is another case of equality because ${\tt true}\, \overline{\oplus}\, {\tt false}={\tt inc}$.
    \item $v_{\Delta_1}(i_2)={\tt true}$, 
    $v_{\Delta_2}(i_2)={\tt true}$ and $v_{\Delta}(i_2)={\tt inc}$.\\
    This is a case where equality does not hold because ${\tt true}\, \overline{\oplus}\, {\tt true}\ne {\tt inc}$. Notice however that ${\tt true}\,\triangleleft\,{\tt inc}$ holds.\hfill$\Box$
\end{itemize}
\end{example}
\section{Consistent Query Answering}\label{sec:query}
In this section, considering true tuples and false tuples only ({\em i.e.,} forgetting about false tuples), we address the important problem of {\em consistent query answering}. We first provide a brief review  of the abundant related literature, and then, we show that our approach provides new insights in the problem of consistent query answering. Moreover, we also argue that in our approach, the `quality' of such consistent answers can be assessed, based on the notion of tuple truth value. However, this issue lies out of the scope of the present paper, and will be the subject of further research in the next future.
\subsection{Related Work}
The problem of query answering in presence of inconsistencies has motivated important research efforts during the past two decades and is still the subject of current research. As mentioned in the introductory section, the most popular approaches in the literature are based on the notion of `repair', a repair of ${\cal D}$ being intuitively {\em a consistent database ${\cal R}$ `as close as possible' to ${\cal D}$}; and an answer to a query $Q$ is consistent if it is present in {\em every repair} ${\cal R}$ of ${\cal D}$. 

However, it has been recognized that generating {\em all} repairs is difficult to implement - if not unfeasible. This is a well known problem in practice which explains, for instance, why data cleansing is a very important but tedious task in the management of databases and data warehouses \cite{RahmD00}. This issue has been thoroughly investigated in \cite{LivshitsKR20}, where it has been shown that computing repairs of a given relational table in the presence of functional dependencies is either polynomial or APX-complete{\footnote{Roughly, APX is the set of NP optimization problems that allow polynomial-time approximation algorithms (source: Wikipedia).}}, depending on the form of the functional dependencies. The reader is referred to \cite{AfratiK09} for theoretical results on the complexity of testing whether ${\cal R}$ is a repair of ${\cal D}$, when considering a more generic context than we do in this work (more than one table and constraints other than functional dependencies). A Prolog based approach for the generation of repairs can be found in \cite{ArieliDNB06}. 

Dealing with repairs without generating them is thus an important issue, also known as {\em Consistent Query Answering in Inconsistent Databases}. One of the first works in this area is \cite{Bry97} and the problem has since been addressed in the context of various database models (mainly the relational model or deductive database models) and under various types of constraints (first order constraints, key constraints, key foreign-key constraints). Seminal papers in this area are \cite{ArenasBC99} and \cite {Wijsen09}, while an overview of works in this area can be found in \cite{Bertossi2011}. 

The problem considered in all these works can be stated as follows: Given a database ${\cal D}$ with integrity constraints ${\cal IC}$, assume that ${\cal D}$ is inconsistent with respect to ${\cal IC}$. Under this assumption, given a query $Q$ against ${\cal D}$, what is the {\em consistent answer} to $Q$? 
The usual approach to alleviate the impact of inconsistent data on the answers to a query is to consider that an answer to $Q$ is consistent if it is present in {\em every repair} ${\cal R}$ of ${\cal D}$. 

Complexity results regarding the computation of the consistent answer have been widely studied in \cite{CaliLR03}. For example one important case is when ${\cal IC}$ consists in having one key constraint per database relation and $Q$ is a conjunctive query containing no self-join ({\em i.e.,} no join of a relation with itself). In this case computing the consistent answer is polynomial whereas if self-joins occur then the problem is co-NP-complete.

Another important problem in considering repairs is that there are many ways of defining the notion of repair. This is so because there are many ways of defining a distance between two database instances, and there is no consensus as to the `best' definition of distance. Although the distance based on symmetric difference seems to be the most popular, other distances exist as well based for example on sub-sets, on cardinality, on updates or on homomorphism \cite{Wijsen05}. Notice in this respect that the results in \cite{LivshitsKR20} are set for two distances: one based on sub-sets and one based on updates. 

\subsection{Consistent Query Answering in our Approach}
In our work we do {\em not} use any notion of repair, thus we avoid the above problem of choosing among all possible ways of defining repairs. Instead, we use set theoretic semantics for tuples and functional dependencies that allow us to associate each tuple with one truth value among true, false, inconsistent or unknown.

In what follows, we outline the process of consistent query answering in our approach, and then compare it to the approaches based on repairs. In doing so we follow the intuition of the repairs-approach where an answer to a query is consistent if it is present in every repair; and we transpose it in our approach by considering that a tuple is in the consistent answer to the query if its truth value is {\tt true} in the sense of our model. 

\smallskip
As usual when dealing with a single table with nulls, a query $Q$ is an SQL-like expression of one of the following two forms:

\smallskip\centerline{
$Q:$ {\sf SELECT} $X$ \qquad or \qquad $Q:$ {\sf SELECT} $X$ {\sf WHERE $\Gamma$}}

\smallskip\noindent
In either of these forms, $X$ is an attribute list seen as a relation schema, and in the second form, the {\sf WHERE} clause specifies a selection condition $\Gamma$. It should thus be clear that, as in SQL, the where clause in a query is optional. The generic form of a query $Q$ is denoted by $Q: {\sf SELECT}~X~{\sf [WHERE~}\Gamma{\sf ]}$.

A selection condition $\Gamma$ is a well formed formula involving the usual connectors $\neg$, $\vee$ and $\wedge$ and built up from atomic boolean comparisons of one of the forms $A \,\theta\, a$ or $A \,\theta\, A'$, where $\theta$ is a comparison predicate, $A$ and $A'$ are attributes in $U$ whose domain elements are comparable through $\theta$, and $a$ is in $dom(A)$.

Moreover, a tuple $t$ satisfies $A\,\theta\, a$ if $A$ is in $sch(t)$ and if $t.A \,\theta\, a$ holds, and $t$ satisfies $A \,\theta\, A'$ if $A$ and $A'$ are  in $sch(t)$ and if $t.A \,\theta\, t.A'$ holds. Based on this, determining whether $t$ satisfies $\Gamma$ follows the rules usual in First Order Logic regarding connectors. For instance, referring to our introductory example, the tuple $t=(k,m)$ such that $sch(t)=KM$ satisfies the conditions $(K=k)$ and $(M=m \vee C=c')$ but does not satisfy the condition $(M=K)$, assuming that $m$ and $k$ are comparable but distinct constants. 

\smallskip
Given  $\Delta = (D, {\cal FD})$, the {\em answer to $Q$ in $\Delta$} is the set of the restrictions to $X$ of all tuples $t$ in $D^*$ such that $X \subseteq sch(t)$ and such that $t$ satisfies $\Gamma$, when present in $Q$. It follows that answers to queries contain only tuples without nulls.

Now, roughly speaking, the {\em consistent answer} to $Q$ is the set of all {\em true} tuples defined over $X$ that satisfy the condition in $Q$. However, as the following example shows, this rough definition should be carefully stated in particular with regard to the functional dependencies to be taken into account for tuple truth value.
\begin{example}\label{ex:new-ex-query}
In the context of our introductory example, let $\Delta =(D, {\cal FD})$ where  ${\cal FD}=\{Id \to K, Id \to C\}$ and  where $D$ is displayed in Figure~\ref{fig:ex-tables-intro}. As seen in Example~\ref{ex:runex-chase}, Algorithm~\ref{algo:chase} returns   $D^*$ as  shown below and $inc({\cal FD})=\{inc(Id \to K), inc(Id \to C)\}$ where $inc(Id \to K)=\emptyset$ and $inc(Id \to C)=\{i_2\}$.

\begin{center}
{\small
\begin{tabular}{c|llll}
$D^*$&$Id$&$K$&$M$&$C$\\
\hline
&$i_1$&$k$&$m$&$c$\\
&$i_1$&$k$&$m'$&$c$\\
&$i_2$&$k'$&$m'$&$c$\\
&$i_2$&$k'$&$m''$&$c$\\
&$i_2$&$k'$&$m'$&$c'$\\
&$i_2$&$k'$&$m''$&$c'$\\
&$i_3$&$k'$&$m$&\\
\end{tabular}
}
\end{center}
We also recall from Example~\ref{ex:truth-value} that ${\sf Inc}(\Delta)$ is defined by:

\smallskip\noindent
\begin{tabular}{rl}
${\sf Inc}(\Delta)=$&$\{t~|~i_2 \sqsubseteq t \sqsubseteq (i_2,k',m',c)\} \cup \{t~|~i_2 \sqsubseteq t \sqsubseteq (i_2,k',m'',c)\} \, \cup$\\
&$\{t~|~i_2 \sqsubseteq t \sqsubseteq (i_2,k',m',c')\} \cup \{t~|~i_2 \sqsubseteq t \sqsubseteq (i_2,k',m'',c')\}$\\
\end{tabular}

\smallskip\noindent
Let $Q_1$ and $Q_2$ be two queries (without conditions) as defined below:

\smallskip\centerline{
\begin{tabular}{lll}
$Q_1 : {\sf SELECT}$ $Id,K,C$
&and&
$Q_2 : {\sf SELECT}$ $Id,K,M$
\\
\end{tabular}}

\smallskip\noindent
Projecting the tuples in $D^*$ over  the attributes $Id$, $K$, $C$ for $Q_1$ and over $Id$, $K$, $M$ for $Q_2$ produces the tables $\Pi_1$ and $\Pi_2$ shown below.

\begin{center}
{\small
\begin{tabular}{c|lll}
$\Pi_1$&$Id$&$K$&$C$\\
\hline
&$i_1$&$k$&$c$\\
&$i_2$&$k'$&$c$\\
&$i_2$&$k'$&$c'$\\
\end{tabular}
\qquad\qquad
\begin{tabular}{c|lll}
$\Pi_2$&$Id$&$K$&$M$\\
\hline
&$i_1$&$k$&$m$\\
&$i_1$&$k$&$m'$\\
&$i_2$&$k'$&$m'$\\
&$i_2$&$k'$&$m''$\\
&$i_3$&$k'$&$m$\\
\end{tabular}
}
\end{center}

\noindent
Since in these two tables, the tuples whose $Id$-value is $i_2$, are inconsistent in $\Delta$, it seems justified to exclude them from any consistent answer. In other words, according to this intuition, the expected consistent answers to $Q_1$ and $Q_2$ are respectively $\{(i_1,k,c)\}$ and $\{(i_1,k,m),(i_1,k,m'), (i_3, k',m)\}$. 

We explain below why it makes sense to exclude the two tuples in the case of $Q_1$, whereas the removal in the case of $Q_2$ is debatable.
\begin{enumerate}
\item
Regarding $Q_1$, the tuples $(i_2,k',c)$ and $(i_2, k', c')$  in $\Pi_1$ clearly violate $Id \to C$ from ${\cal FD}$, and thus can not occur in the {\em consistent} answer to $Q_1$.
\item
Regarding $Q_2$ however, no functional dependency is violated by the tuples in $\Pi_2$, and thus, there is no reason for removing any of them when producing the {\em consistent} answer to $Q_2$.
\end{enumerate}
Another way of explaining this situation is to notice that, in $D^*$, the only non satisfied functional dependency is $Id \to C$ and that
\begin{enumerate}
\item
attributes $Id$ and $C$ occur in the {\sf SELECT} clause of $Q_1$, making it necessary to check functional dependency satisfaction;
\item
attribute $C$ does not occur in the {\sf SELECT} clause of $Q_2$, implying that checking functional dependency satisfaction makes no sense.
\end{enumerate}
Another important point to take into account is the impact of selection conditions on tuple truth value in the answer to a query. To illustrate this point, first notice that, when considering the query $Q_1$ the only functional dependency to be checked is  $Id \to C$, with respect to which the table $\Pi_1$ shows inconsistencies regarding $i_2$. However, let now $Q'_1$ be the query defined by:

\smallskip\centerline{
$Q'_1 : {\sf SELECT}$ $Id,K,C$ {\sf WHERE} $C=c'$}

\smallskip\noindent
Only the fifth and sixth tuples in $D^*$ satisfy the selection condition and thus, the only possible tuple in the consistent answer to $Q'_1$ is $(i_2, k', c')$, which alone, trivially satisfies  the functional dependency $Id \to C$.

However, the consistency of the answer to $Q'_1$ may seem counter-intuitive, since the tuple $(i_2, k', c')$ is seen as {\em inconsistent} in the answer to $Q_1$, where the same attributes are involved. To cope with this counter-intuitive situation, we rather consider that the consistent answer of $Q'_1$ is {\em empty}, {\em i.e.,} that consistency has to be checked independently from selection conditions, based only on the functional dependencies involving only attributes from the {\sf SELECT} clause in the query.

In what follows, we provide the formalism and the definitions to account for these remarks.\hfill$\Box$
\end{example}
Given a table $D$ over $U$, a subset $X$ of $U$ and a selection condition $\Gamma$, we denote by $\sigma_{\Gamma}(D)$, $\pi_X(D)$ and $\pi_X({\cal FD})$ the following sets:
\begin{itemize}
\item
$\sigma_{\Gamma}(D)$ is the set of all tuples $t$ in $D$ such that $t$ satisfies $\Gamma$.
\item
$\pi_X(D)$ is the set of the restrictions to $X$ of all tuples in $D$ whose schema contains $X$; that is $\pi_X(D) = \{t~|~(\exists q \in D)(X \subseteq sch(q),~t = q.X)\}$.
\item
$\pi_X({\cal FD})$ is the set of all functional dependencies that involve attributes in $X$ only; that is $\pi_X({\cal FD}) = \{(Y\to B) \in {\cal FD}~|~YB \subseteq X\}$.
\end{itemize}
These notation are used in the following definition where the notion of {\em consistent answer to a query} is introduced.
\begin{definition}\label{def:answer}
Given $\Delta=(D, {\cal FD})$ and $Q: {\sf SELECT}$ $X$ {\sf [WHERE $\Gamma$]}, let $\Delta_X$ be defined by  $\Delta_X=(\pi_X(D^*), \pi_X({\cal FD}))$.

The {\em answer to $Q$ in $\Delta$}, denoted by $ans_\Delta(Q)$, is the set  $\pi_X(\sigma_\Gamma(D^*))$. Moreover, for every tuple $x$ in $ans_\Delta(Q)$, the {\em truth value of $x$ in $ans_\Delta(Q)$} is defined by $v_{\Delta_X}(x)$.

The {\em consistent answer} to $Q$ in $\Delta$, denoted by $ans^+_\Delta(Q)$, is the set of all tuples $x$ in $ans_\Delta(Q)$ such that $v_{\Delta_X}(x)={\tt true}$.
\end{definition}
It is important to notice that, according to Definition~\ref{def:answer}, given $\Delta$ and $Q$, {\em two distinct} truth values may be given to a tuple $t$, namely, its truth value in $\Delta$, {\em i.e.,} $v _\Delta(t)$, and its truth value in $\Delta_X$, {\em i.e.,} $v_{\Delta_X}(t)$. Since these truth values are not determined using the {\em same} set of functional dependencies, they might be distinct.

Referring to Example~\ref{ex:new-ex-query}, based on the notation introduced in Definition~\ref{def:answer}, 
for $X_1= Id\,K\,C$, $\pi_{X_1}({\cal FD})={\cal FD}$, and so, $\Delta_{X_1} =(\Pi_1, {\cal FD})$. In this case, for every tuple $x$ over $X_1$, $v_{\Delta_{X_1}}(x)=v_\Delta(x)$.
On the other hand, for $X_2=Id\,K\,M$, $\pi_{X_2}({\cal FD})=\{Id \to K\}$, and so, $\Delta_{X_2} =(\Pi_2, \{Id \to K\})$. Since $\Pi_2$ satisfies $Id \to K$, for $x=(i_2,k',m')$, $v_{\Delta_{X_2}}(x)={\tt true}$. However, as $x$ is a super-tuple of $i_2$, we have $v_\Delta(x)={\tt inc}$, showing that $v_{\Delta_{X_2}}(x) \ne v_\Delta(x)$.

{\small
\algsetup{indent=1.5em}
\begin{algorithm}[t]
\caption{Consistent answer $ans^+_\Delta(Q)$\label{algo:answer}}
\begin{algorithmic}[1]
\REQUIRE
A query $Q: {\sf SELECT}$ $X$ {\sf [WHERE $\Gamma$]}, $\Delta^* = (D^*, {\cal FD})$ and  $inc({\cal FD})$

\ENSURE The set $ans^+_\Delta(Q)$
\STATE{$ans^+_\Delta(Q):= \emptyset$}
\FORALL{$t$ in $D^*$}
    \IF{$sch(t)$ contains all attributes in $X$}
        \IF{for every $Y \to B$ in $\pi_X({\cal FD})$, $t.Y$ is not in $inc(Y \to B)$}\label{line:test0}
            \IF{$t$ satisfies $\Gamma$}\label{line:test-Gamma}
		        \STATE{\COMMENT{This test always succeeds if $Q$ involves no selection condition}}
        		\STATE{$ans^+_\Delta(Q) := ans^+_\Delta(Q) \cup \{t.X\}$}
		    \ENDIF
	    \ENDIF
    \ENDIF
\ENDFOR
\RETURN{$ans^+_\Delta(Q)$}
\end{algorithmic}
\end{algorithm}
}

\smallskip
The following proposition shows that $ans^+_\Delta(Q)$ is computed from $D^*$ and $inc({\cal FD})$, using  Algorithm~\ref{algo:answer}.

\begin{proposition}\label{prop:cons-ans}
Given $\Delta=(D, {\cal FD})$ and~$Q: {\sf SELECT}$ $X$ {\sf [WHERE $\Gamma$]}, Algorithm~\ref{algo:answer} correctly computes $ans^+_\Delta(Q)$.
\end{proposition}
\begin{proof}
In this proof, denoting by $ans$ the output of Algorithm~\ref{algo:answer}, we prove that $ans=ans^+_\Delta(Q)$. To prove that $ans \subseteq ans^+_\Delta(Q)$, we notice that, by Algorithm~\ref{algo:answer}, every tuple $x$ in $ans$ $x$ is the projection over $X$ of a tuple $t$ in $D^*$ satisfying $\Gamma$. Thus, $x$ belongs to $\pi_X(\sigma_\Gamma(D^*))$, that is to $ans_\Delta(Q)$. Moreover, since for every $t$ in $D^*$ such that $t.X$ is in $ans$ and every $Y \to B$ in $\pi_X({\cal FD})$, $t.Y$ is not in $inc(Y \to B)$, it holds that $v_{\Delta_X}(x)={\tt true}$. It thus follows that $x$ is in $ans^+_\Delta(Q)$.

Conversely, assuming that $x$ is in $ans^+_\Delta(Q)$ implies that $x$ is in $ans_\Delta(Q)$. Hence, $D^*$ contains a tuple $t$ that satisfies $\Gamma$ and $t.X=x$, meaning that $sch(t)$ contains $X$ and that  $t$ satisfies the if-condition on line~\ref{line:test-Gamma} in Algorithm~\ref{algo:answer}. Moreover, since we also have $v_{\Delta_X}(x)={\tt true}$, for every $Y \to B$ in $\pi_X({\cal FD})$, $t.Y$ cannot be in $inc(Y \to B)$. This shows that the if-condition on line~\ref{line:test0} in Algorithm~\ref{algo:answer} is satisfied, and thus that $x$ belongs to $ans$, which completes the proof.\hfill$\Box$
\end{proof}
Regarding complexity, Proposition~\ref{prop:cons-ans} shows that, assuming that $D^*$ has been computed, the computation of the consistent answer is {\em linear} in the size of $D^*$.

If we assume moreover that ${\sf Inc}(\Delta)$ has also been computed, labelling each tuple in $ans^+_\Delta(Q)$ by its truth value in $\Delta$ is an option to investigate, because it has been seen from Definition~\ref{def:answer} that the truth value of a tuple $t$ in $\Delta$, {\em i.e.,} $v_\Delta(t)$, may be different than the truth value of $t$ in $ans_\Delta(Q)$, {\em i.e.,} $v_{\Delta_X}(t)$.

Knowing that a tuple in the consistent answer, thus having truth value {\tt true} in this answer, has truth value {\tt inc} in the database it comes from, may indeed be relevant in case the user is interested in data quality, as is the case when dealing with data lakes \cite{MEDES}.  Investigating further issues related to query answering in our approach, including issues related to data quality is the subject of future work.
\begin{example}\label{ex:new-ex-query-bis}
Running Algorithm~\ref{algo:answer} with the queries $Q_1$, $Q'_1$ and $Q_2$ as in Example~\ref{ex:new-ex-query}, returns $ans^+_\Delta(Q_1)=\{(i_1,k,c)\}$, $ans^+_\Delta(Q'_1)=\emptyset$ and $ans^+_\Delta(Q_2)=\{(i_1,k,m),$ $(i_1,k,m'),$ $(i_2,k',m'),$ $(i_2,k',m''),$ $(i_3, k',m)\}$, as expected.

As earlier noticed regarding $ans^+_\Delta(Q_2)$, for $x=(i_2,k',m')$ or $x=(i_2,k',m'')$, we have $v_\Delta(x) \ne v_{\Delta_{X_2}}(x)$. In this case, smart users could find it relevant to be informed of this situation, which can be done by labelling the two tuples $(i_2,k',m')$ and $(i_2,k',m'')$ by ${\tt inc}$, that is, their truth value in $\Delta$. We notice that this piece of information cannot be provided by any of the existing approaches.

\smallskip
Considering now the query $Q_3 : {\sf SELECT}$ $M, C$ {\sf WHERE} $K=k'$, Algorithm~\ref{algo:answer} discards the first two tuples of $D^*$ (because their $K$-value is not equal to $k'$), and also the last tuple of $D^*$ (as this tuple has no $C$-value). When processing the remaining four tuples in $D^*$, no functional dependency has to be taken care of, and so, we obtain $ans^+_\Delta(Q_3)=\{(m', c),$ $(m', c'),$ $(m'',c),$ $(m'',c')\}$.
\hfill$\Box$
\end{example}
\subsection{Comparison with Repair-Based Approaches} 
Comparing our approach with approaches to consistent query answering from the literature, we point out that when constraints are functional dependencies only, as  in our approach, repairs are defined using set-theoretic inclusion as follows.
\begin{definition}\label{def:repair}
Given $\Delta =(D, {\cal FD})$, denoting by $D^*$ the chased table associated with $D$, a {\em repair} of $\Delta$ is a table $R$ over $U$ such that:
(1) $R \subseteq D^*$,
(2) $R$ satisfies ${\cal FD}$, and
(3) $R$ is maximal among the sets satisfying $(1)$ and $(2)$.
\end{definition}

\noindent
We notice that in the above definition, inclusion is understood in its strict set-theoretic meaning, disregarding the presence of nulls in the tuples. For example $\{ab, a'bc\} \subseteq \{abc, a'bc\}$ does not hold whereas $\{ab, a'bc\} \subseteq \{ab, abc, a'bc\}$ does.

Repairs of $\Delta$ can be generated based on the tuples stored in $inc({\cal FD})$ according to the following algorithm:

\smallskip
{\small
$R := D^*$

{\bf for all} $X \to A$ in ${\cal FD}$ {\bf do}

\indent\indent
{\bf for all} $x$ in $inc(X \to A)$ {\bf do}

\hspace{1cm}
choose an $A$-value $a$ among all $\alpha$ such that  $x \alpha$ occurs in $D^*$

\hspace{1cm}
$R := R \setminus \{q~|~XA \subseteq sch(q), q.X=x, q.A \ne a\}$

{\bf return} $R$
}

\smallskip\noindent
Indeed, based on Definition~\ref{def:repair},  $R$ as computed above is a repair because: (1) $R \subseteq D^*$ clearly holds, (2) $R$ satisfies ${\cal FD}$ holds since for every  $X \to A$ in ${\cal FD}$, there exist $q$ and $q'$ in $R$ such that $q.X=q'.X$ and $q.A \ne q'.A$, and (3) $R$ is maximal because inserting any of the removed tuples leads to violation of a functional dependency.

\smallskip
Given a query $Q: {\sf SELECT}$ $X$ {\sf [WHERE $\Gamma$]}, denoting by $Rep(\Delta)$ the set of all repairs of $\Delta$, the {\em consistent answer to $Q$ based on repairs} can be formally defined in the following two ways:
\begin{enumerate}
\item
$ans^\downarrow_\Delta(Q)= \pi_X\left(\bigcap_{R \in Rep(\Delta)} \sigma_\Gamma (R)\right)$.
\item
$ans^\uparrow_\Delta(Q)= \bigcap_{R \in Rep(\Delta)} \pi_X(\sigma_\Gamma (R))$.
\end{enumerate}
Intuitively, $ans^\downarrow_\Delta(Q)$ is obtained by evaluating the query against the intersection of all repairs, whereas $ans^\uparrow_\Delta(Q)$ is obtained by evaluating the query against each repair and by taking the intersection of all these answers.

\begin{example}\label{ex:new-ex-query-rep}
Computing the repairs of $D^*$ as shown in Example~\ref{ex:new-ex-query} produces the tables $R_1$ and $R_2$ shown below.

\begin{center}
{\small
\begin{tabular}{c|llll}
$R_1$&$Id$&$K$&$M$&$C$\\
\hline
&$i_1$&$k$&$m$&$c$\\
&$i_1$&$k$&$m'$&$c$\\
&$i_2$&$k'$&$m'$&$c$\\
&$i_2$&$k'$&$m''$&$c$\\
&$i_3$&$k'$&$m$&\\
\end{tabular}
\qquad\qquad
\begin{tabular}{c|llll}
$R_2$&$Id$&$K$&$M$&$C$\\
\hline
&$i_1$&$k$&$m$&$c$\\
&$i_1$&$k$&$m'$&$c$\\
&$i_2$&$k'$&$m'$&$c'$\\
&$i_2$&$k'$&$m''$&$c'$\\
&$i_3$&$k'$&$m$&\\
\end{tabular}
}
\end{center}
Thus, regarding the queries $Q_1$, $Q'_1$, $Q_2$ and $Q_3$ of Example~\ref{ex:new-ex-query}, we have:
\begin{itemize}
\item $ans^\downarrow_\Delta(Q_1) = \{(i_1,k,c)\}$~;~$ans^\uparrow_\Delta(Q_1) = \{(i_1,k,c)\}$
\item $ans^\downarrow_\Delta(Q'_1) = \emptyset$~;~$ans^\uparrow_\Delta(Q'_1) = \emptyset$
\item $ans^\downarrow_\Delta(Q_2) =\{(i_1,k,m), (i_1, k, m'), (i_3,k',m)\}$~;\\
$ans^\uparrow_\Delta(Q_2) =\{(i_1,k,m), (i_1, k, m'), (i_2,k',m'), (i_2,k',m''), (i_3,k',m)\}$
\item $ans^\downarrow_\Delta(Q_3) = \emptyset$~;~$ans^\uparrow_\Delta(Q_3) = \emptyset$
\end{itemize}
It should be noticed that computing all repairs before computing the answers is not realistic in practice. In what follows, we provide an efficient algorithm to compute these answers and we prove that they are always `smaller' with respect to set theoretic inclusion than the answers as defined in Definition~\ref{def:answer}.
\hfill$\Box$
\end{example}

{\small
\algsetup{indent=1.5em}
\begin{algorithm}[t]
\caption{Repair-based consistent answers  $ans^\downarrow_\Delta(Q)$, $ans^\uparrow_\Delta(Q)$\label{algo:answer-repair}}
\begin{algorithmic}[1]
\REQUIRE
A query $Q: {\sf SELECT}$ $X$ {\sf [WHERE $\Gamma$]}, $\Delta^* = (D^*, {\cal FD})$ and  $inc({\cal FD})$

\ENSURE The sets $ans^\downarrow(Q)$ and $ans^\uparrow(Q)$ 

\STATE{$ans^\downarrow(Q) := \emptyset$ ; $ans^\uparrow (Q):= \emptyset$}
\FORALL{$t$ in $D^*$}
    \IF{$sch(t)$ contains all attributes in $X$}
    	\IF{$t$ satisfies $\Gamma$} 
		\STATE{\COMMENT{This test always succeeds if $Q$ involves no selection condition}}
		\IF{for every $Y \to B$ in ${\cal FD}$ such that $YB \subseteq sch(t)$, $t.Y$ is not in $inc(Y \to B)$}\label{line:test1}
        			\STATE{$ans^\downarrow(Q) := ans^\downarrow(Q) \cup \{t.X\}$}
		\ENDIF
		\IF{for every $Y \to B$ in ${\cal FD}$ such that $YB \subseteq sch(t)$ and $B \in X$, $t.Y$ is not in $inc(Y \to B)$}\label{line:test2}
        			\STATE{$ans^\uparrow(Q) := ans^\uparrow(Q) \cup \{t.X\}$}
		\ENDIF
	\ENDIF
    \ENDIF
\ENDFOR
\RETURN{$ans^\downarrow(Q)$, $ans^\uparrow(Q)$}
\end{algorithmic}
\end{algorithm}
}

\noindent
The following proposition deals with the computation of $ans^\downarrow_\Delta(Q)$ and of $ans^\uparrow_\Delta(Q)$, and compares these answers with $ans^+_\Delta(Q)$.
\begin{proposition}\label{prop:cons-answers}
Given $\Delta=(D, {\cal FD})$ and a query {\rm $Q: {\sf SELECT}$ $X$ {\sf [WHERE $\Gamma$]}}, Algorithm~\ref{algo:answer-repair} correctly computes $ans^\downarrow_\Delta(Q)$ and $ans^\uparrow_\Delta(Q)$.
Moreover, the following holds:
$ans^\downarrow_\Delta(Q) \subseteq ans^\uparrow_\Delta(Q) \subseteq ans^+_\Delta(Q)$.
\end{proposition}
\begin{proof}
See Appendix~\ref{append:proof-cons-answers}.\hfill$\Box$
\end{proof}
To illustrate the inclusions in Proposition~\ref{prop:cons-answers}, it can be seen from Example~\ref{ex:new-ex-query-bis} and Example~\ref{ex:new-ex-query-rep} that: 

\smallskip\noindent
$-$ $ans^\downarrow_\Delta(Q_1) = ans^\uparrow_\Delta(Q_1)=ans^+_\Delta(Q_1)$;
\\
$-$ $ans^\downarrow_\Delta(Q'_1) = ans^\uparrow_\Delta(Q'_1)=ans^+_\Delta(Q'_1)$;
\\
$-$ $ans^\downarrow_\Delta(Q_2) \subset ans^\uparrow_\Delta(Q_2)$ and $ans^\uparrow_\Delta(Q_2)=ans^+_\Delta(Q_2)$;
\\
$-$  $ans^\downarrow_\Delta(Q_3)= ans^\uparrow_\Delta(Q_3)$ and $ans^\uparrow_\Delta(Q_3) \subset ans^+_\Delta(Q_3)$.

\smallskip\noindent
Regarding complexity, is important to note that, if the chased table $D^*$ is available then any of the three ways to compute consistent query answers is {\em linear} in the size of $D^*$.
Moreover, when providing any of these consistent answers, our approach allows for pointing to the user possible problematic tuples, namely those tuples that are {\em inconsistent in $\Delta$, although not inconsistent in the answer.}
\section{Concluding Remarks}\label{sec:conclusion}
In this paper we have introduced a novel approach to handle inconsistencies in a table with nulls and functional dependencies. Our approach uses set theoretic semantics and relies on an extended version of the well known chase procedure to associate every possible tuple with one of the four truth values true, false, inconsistent and unknown. Moreover, we have seen that true and inconsistent tuples can be computed in time polynomial in the size of the input table. 
We have also seen that our approach applies to consistent query answering and we have shown that it provides larger answers than the repair-based approaches.

Building upon these results, we currently pursue four lines of research: $(1)$ applying our approach to the particular but important case of key-foreign key constraints in the context of a star schema or a snow-flake schema; $(2)$ designing incremental algorithms to improve performance in case of updates, $(3)$ extending our approach to constraints other than functional dependencies, such as inclusion dependencies as done in \cite{BravoB06}, (4) investigating the issue of data quality in the framework of our approach, and $(5)$ extending our approach to account for the presence of tuples declared as {\em false}.
\section*{Declarations}

{\bf Author contributions:} The two authors contributed to the study, conception and design. Both read and approved the submitted manuscript.

\smallskip\noindent
{\bf Funding:} No funds, grants, or other support was received for conducting this study.

\smallskip\noindent
{\bf Financial interests:} N/A.

\smallskip\noindent
{\bf Non-financial interests:} N/A.

\smallskip\noindent
{\bf Data availability:} Data sharing is not applicable to this article as no datasets were generated or analyzed during the current study.


\appendix
\section{Proof of Lemma~\ref{lemma:least-model}}\label{append:least-model}
{\bf Lemma~\ref{lemma:least-model}.}
{\em
For every $\Delta = (D, {\cal FD})$, the sequence $\left(\mu_i\right)_{i \geq 0}$ has a unique limit $\mu^*$ that satisfies that  $\mu^* \models \Delta$.
Moreover:
\begin{enumerate}
\item
For all $a_1$ and $a_2$ in the {\em same attribute domain} $dom(A)$, if $\mu^*(a_1) \cap \mu^*(a_2) \ne \emptyset$ then there exist $X \to A$ in ${\cal FD}$ and $x$ over $X$ such that $\mu^*(x) \ne \emptyset$ and $\mu^*(x) \subseteq \mu^*(a_1) \cap \mu^*(a_2)$.
\item
For all $\alpha$ and $\beta$, $\Delta \vdash (\alpha \sqcap \beta)$ holds if and only if $\mu^*(\alpha) \cap \mu^*(\beta) \ne \emptyset$ holds.
\end{enumerate}
}
\begin{proof}
We recall that the sequence $\left(\mu_i\right)_{i \geq 0}$ is defined by the following steps:
\begin{enumerate}
\item
For every $t$ in $D$, assign a `fresh' integer $id(t)$ to $t$;
\item
Let $\mu_0$ be the mapping defined for every domain constant $a$ by:\\ $\mu_0(a) =\{id(t)~|~t \in D \mbox{ and }a \sqsubseteq t\}$;
\item
While there exists $X \to A$ in ${\cal FD}$, $x$ over $X$ and $a$ in $dom(A)$ such that $\mu(xa) \ne \emptyset$ and $\mu(x) \not\subseteq \mu(a)$, define $\mu_{i+1}$ by: $\mu_{i+1}(a) = \mu_i(a) \cup \mu_i(x)$ and $\mu_{i+1}(\alpha) = \mu_i(\alpha)$ for any other constant $\alpha$.
\end{enumerate}
The sequence $\left(\mu_i\right)_{i \geq 0}$ is increasing in the sense that for every $\alpha$, $\mu_i(\alpha) \subseteq \mu_{i+1}(\alpha)$, and bounded in the sense that for every $\alpha$, $\mu_i(\alpha)\subseteq \{id(t)~|~t\in \Delta\}$. Hence the sequence has a unique limit. Moreover, for every $t$ in $D$, $\mu^*(t) \ne \emptyset$ holds because $id(t)$ always belongs to $\mu^*(t)$, and $\mu^* \models {\cal FD}$, because otherwise $\mu^*$ would not be the limit of the sequence. Therefore $\mu^* \models \Delta$, which shows the first part of the lemma.

\smallskip\noindent
$(1)$ Regarding the first item in the second part of the lemma, we first notice that by definition of $\mu_0$, we have $\mu_0(a_1) \cap \mu_0(a_2) = \emptyset$, because it is not possible that a tuple in $D$ has two distinct values over an attribute.

Since we assume that $\mu^*(a_1) \cap \mu^*(a_2) \ne \emptyset$, there exists $i_0 \geq 0$ such that $\mu_{i_0}(a_1) \cap \mu_{i_0}(a_2) = \emptyset$ and $\mu_{i_0 +1}(a_1) \cap \mu_{i_0 +1}(a_2) \ne \emptyset$. By definition of the sequence $\left(\mu_i \right)_{i \geq 0}$, for $j=1,2$, $\mu_{i_0+1}(a_j) =\mu_{i_0}(a_j) \cup M(a_j)$ where $M(a_j)$ is the union of all $\mu_{i_0}(x_j)$ such that $X_j \to A$ is in ${\cal FD}$,  $\mu_{i_0}(x_j) \cap \mu_{i_0}(a_j) \ne \emptyset$ and $\mu_{i_0}(x_j) \not\subseteq  \mu_{i_0}(a_j)$. Hence, 

\noindent
\begin{tabular}{rl}
$\mu_{i_0+1}(a_1) \cap \mu_{i_0+1}(a_2)$&$=(\mu_{i_0}(a_1) \cup M(a_1))\cap (\mu_{i_0}(a_2) \cup M(a_2))$\\
&$= (\mu_{i_0}(a_1) \cap \mu_{i_0}(a_2))\cup  (\mu_{i_0}(a_1) \cap M(a_2))~\cup ~\qquad$\\
&\hfill $(M(a_1)\cap \mu_{i_0}(a_2))\cup (M(a_1)\cap M(a_2))$\\
\end{tabular}

\noindent
Since $\mu_{i_0+1}(a_1) \cap \mu_{i_0+1}(a_2) \ne \emptyset$, at least one of the four terms of the above union is not empty. But since $\mu_{i_0}(a_1) \cap \mu_{i_0}(a_2)=\emptyset$, only the last three cases are investigated below.

\smallskip\noindent
$(i)$ If $\mu_{i_0}(a_1) \cap M(a_2)\ne \emptyset$, $M(a_2)$ contains $x_2$ such that $\mu_{i_0}(a_1) \cap \mu_{i_0}(x_2) \ne \emptyset$. Thus, there exists $X_2 \to A$ is in ${\cal FD}$ such that  $X_2 = sch(x_2)$, $\mu_{i_0}(a_1) \cap \mu_{i_0}(x_2) \ne \emptyset$ and $\mu_{i_0}(a_2) \cap \mu_{i_0}(x_2) \ne \emptyset$. Since both $a_1$ and $a_2$ are in $dom(A)$, we have $\mu_{i_0+1}(x_2) \subseteq \mu_{i_0+1}(a_1)$ and $\mu_{i_0+1}(x_2) \subseteq \mu_{i_0+1}(a_2)$. Thus $\mu^*(x_2) \subseteq \mu^*(a_1) \cap \mu^*(a_2)$.

\smallskip\noindent
$(ii)$ If $\mu_{i_0}(a_2) \cap M(a_1)\ne \emptyset$, it can be shown in a similar way that there exist $X_1 \to A$ is in ${\cal FD}$ and $x_1$ over $X_1$ such that $\mu^*(x_1) \subseteq \mu^*(a_1) \cap \mu^*(a_2)$. The proof is omitted.

\smallskip\noindent
$(iii)$  If $M(a_1) \cap M(a_2)\ne \emptyset$, for $j=1,2$, $M(a_j)$ contains $x_j$ such that $\mu_{i_0}(x_1) \cap \mu_{i_0}(x_2) \ne \emptyset$. Thus, for $j=1,2$, there exist $X_j \to A$ in ${\cal FD}$ such that  $X_j=sch(x_j)$, $\mu_{i_0}(x_j) \cap \mu_{i_0}(a_j) \ne \emptyset$ and $\mu_{i_0}(x_1) \cap \mu_{i_0}(x_2) \ne \emptyset$. Hence, $\mu_{i_0+1}(x_j) \subseteq \mu_{i_0+1}(a_j)$, for $j=1,2$ and $\mu_{i_0+1}(x_1) \cap \mu_{i_0+1}(x_2) \ne \emptyset$. It follows that, when computing $\mu_{i_0+2}$, we obtain the additional inclusions $\mu_{i_0+2}(x_1) \subseteq \mu_{i_0+2}(a_2)$ and $\mu_{i_0+2}(x_2) \subseteq \mu_{i_0+2}(a_1)$, which implies that for $j=1,2$, $\mu^*(x_j) \subseteq \mu^*(a_1) \cap \mu^*(a_2)$ holds. This part of the proof is thus complete.

\smallskip\noindent
$(2)$ Regarding the second item in the second part of the lemma, assume first that $\Delta \vdash (\alpha \sqcap \beta)$. Since $\mu^* \models \Delta$, we obviously have that $\mu^*(\alpha) \cap \mu^*(\beta)\ne \emptyset$.

Conversely, assuming that $\mu^*(\alpha) \cap \mu^*(\beta)\ne \emptyset$, we show that $\Delta \vdash (\alpha \sqcap \beta)$, that is, for every $\mu$ such that $\mu \models \Delta$, $\mu(\alpha) \cap \mu(\beta)\ne \emptyset$. The proof is by induction on the steps of the construction of $\mu^*$, assuming $\alpha$ in $dom(A)$ and $\beta$ in $dom(B)$.
\\
$\bullet$ The result holds for $i=0$. Indeed, if $\mu_0(\alpha) \cap \mu_0(\beta)\ne \emptyset$ then there exists $u$ in $D$ such that $\alpha \sqsubseteq u$ and $\beta \sqsubseteq u$. Hence for every $\mu$ such that $\mu \models \Delta$, we have $\mu(u) \ne \emptyset$ and $\mu(u) \subseteq \mu(\alpha) \cap \mu(\beta)$, implying that $\mu(\alpha) \cap \mu(\beta) \ne \emptyset$ holds. 
\\
$\bullet$ For $i_0 >0$, assuming that $\mu_{i_0}$ satisfies that for all $\zeta$ and $\eta$ such that $\mu_{i_0}(\zeta) \cap \mu_{i_0}(\eta) \ne \emptyset$, we have $\mu(\zeta) \cap \mu(\eta) \ne \emptyset$ for every $\mu$ such that $\mu \models \Delta$, we show that the result holds for $\mu_{i_0+1}$. 

Indeed, let $i_0$ such that $\mu_{i_0}(\alpha) \cap \mu_{i_0}(\beta)=\emptyset$ and $\mu_{i_0+1}(\alpha) \cap \mu_{i_0+1}(\beta) \ne \emptyset$. By definition of the sequence $(\mu_i)_{i \geq 0}$, and as in $(1)$ just above, $\mu_{i_0+1}(\alpha) =\mu_{i_0}(\alpha) \cup M(\alpha)$ where $M(\alpha)$ is the union of all $\mu_{i_0}(x)$ such that $X \to A$ is in ${\cal FD}$, $\mu_{i_0}(x) \cap \mu_{i_0}(\alpha) \ne \emptyset$ and $\mu_{i_0}(x) \not\subseteq  \mu_{i_0}(\alpha)$. Similarly, $\mu_{i_0+1}(\beta) =\mu_{i_0}(\beta) \cup M(\beta)$ where $M(\beta)$ is the union of all $\mu_{i_0}(y)$ such that $Y \to B$ is in ${\cal FD}$, $\mu_{i_0}(y) \cap \mu_{i_0}(\beta) \ne \emptyset$ and $\mu_{i_0}(y) \not\subseteq  \mu_{i_0}(\beta)$. Thus:

\noindent
\begin{tabular}{rl}
$\mu_{i_0+1}(\alpha) \cap \mu_{i_0+1}(\beta)$&$=(\mu_{i_0}(\alpha) \cup M(\alpha))\cap (\mu_{i_0}(\beta) \cup M(\beta))$\\
&$= (\mu_{i_0}(\alpha) \cap \mu_{i_0}(\beta))\cup  (\mu_{i_0}(\alpha) \cap M(\beta))~\cup ~\qquad$\\
&\hfill $(M(\alpha)\cap \mu_{i_0}(\beta))\cup (M(\alpha)\cap M(\beta))$\\
\end{tabular}

\noindent
Since $\mu_{i_0+1}(\alpha) \cap \mu_{i_0+1}(\beta) \ne \emptyset$, at least one of the four terms of the above union is non empty. But since $\mu_{i_0}(\alpha) \cap \mu_{i_0}(\beta)=\emptyset$, only the last three cases are investigated below.

\smallskip\noindent
$(i)$ If $\mu_{i_0}(\alpha) \cap M(\beta)\ne \emptyset$, there exist $Y \to B$ in ${\cal FD}$ and $y$ over $Y$ such that $\mu_{i_0}(\alpha) \cap \mu_{i_0}(y) \ne \emptyset$ and $\mu_{i_0}(\beta) \cap \mu_{i_0}(y) \ne \emptyset$. By our induction hypothesis, for every $\mu$ such that $\mu \models \Delta$, we have $\mu(\alpha) \cap \mu(y) \ne \emptyset$ and $\mu(y) \subseteq \mu(\beta)$, which implies that $\mu(\alpha) \cap \mu(\beta) \ne \emptyset$.

\smallskip\noindent
$(ii)$ If $\mu_{i_0}(\beta) \cap M(\alpha)\ne \emptyset$, the case is similar  to $(i)$ above. The proof is omitted.

\smallskip\noindent
$(iii)$ If $M(\alpha) \cap M(\beta)\ne \emptyset$, there exist $X \to A$ and $Y \to B$ in ${\cal FD}$, $x$ over $X$ and $y$ over $Y$,  such that $\mu_{i_0}(x) \cap \mu_{i_0}(y) \ne \emptyset$, $\mu_{i_0}(\alpha) \cap \mu_{i_0}(x) \ne \emptyset$ and $\mu_{i_0}(\beta) \cap \mu_{i_0}(y) \ne \emptyset$. By our induction hypothesis, for every $\mu$ such that $\mu \models \Delta$, we have $\mu(x) \cap \mu(y) \ne \emptyset$, $\mu(x) \subseteq \mu(\alpha)$ and $\mu(y) \subseteq \mu(\beta)$. Hence,  $\mu(\alpha) \cap \mu(\beta) \ne \emptyset$ also holds in this case, and the proof is complete.\hfill$\Box$
\end{proof}
\section{Proof of Lemma~\ref{lemma:inclusion}}\label{append:lemma-inclusion}
{\bf Lemma~\ref{lemma:inclusion}.}
{\em
Let $\Delta = (D, {\cal FD})$ and $t$ a tuple. Then Algorithm~\ref{algo:closure} computes correctly the closure $t^+$ of $t$. 
}
\begin{proof}
In this proof, we denote by $cl(t)$ the output of Algorithm~\ref{algo:closure}, and we show that $cl(t)=t^+$, that is that $cl(t) \subseteq t^+$ and $t^+ \subseteq cl(t)$ both hold. Before proceeding to these proofs, we draw attention on that for every ${\cal T}$-mapping $\mu$ such that $\mu(t) \ne \emptyset$, $\mu \models \Delta$ if and only if $\mu \models \Delta_t$, where $\Delta_t$ is defined by the statement line~\ref{line:cl-db-assign} in Algorithm~\ref{algo:closure}. Indeed: 
\\
$\bullet$ If $\mu \models \Delta_t$ then for every $q \in D_t$ $\mu(q) \ne \emptyset$ and $\mu \models {\cal FD}$. Since $D \subset D_t$, $\mu(q) \ne \emptyset$ for every $q \in D$, implying that $\mu \models \Delta$ holds.
\\
$\bullet$ Conversely, if $\mu \models \Delta$ then as $\mu(t)$ is supposed to be nonempty, $\mu(q) \ne \emptyset$ for every $q$ in $D_t$. Since $\mu \models {\cal FD}$ holds, $\mu \models \Delta_t$ also holds.

\smallskip
To first prove that $cl(t) \subseteq t^+$, we consider a ${\cal T}$-mapping $\mu$ such that $\mu \models \Delta$, and we prove that $\mu(t) \subseteq \mu(a)$ for every $a$ in $cl(t)$. We first observe that if $\mu(t) =\emptyset$ then $\mu(t) \subseteq \mu(\alpha)$ holds for every constant $\alpha$. Therefore, $\mu(t) \subseteq \mu(a)$ holds.

Now, if $\mu(t) \ne \emptyset$ then $\mu \models \Delta_t$, as shown above. The proof that $\mu(t) \subseteq \mu(a)$ is done by induction on the steps of the execution of Algorithm~\ref{algo:closure}. Denoting by $cl^0$, $cl^1$, $\ldots$ the sequence of the assignments of $cl(t)$ during execution, the following holds for every $a$ in $cl(t)$.
\\
$\bullet$ If $a$ is in $cl^0$ as computed on line~\ref{line:init-cl}, $a$ occurs in $t$. It is thus clear that $\mu(t) \subseteq \mu(a)$.
\\
$\bullet$ We now assume that, for $j \geq 0$, every $\alpha$ in $cl^j$ is such that $\mu(t) \subseteq \mu(\alpha)$ and we show that this holds for $a$ in $cl^{j+1}$ but not in $cl^j$. In this case, according to the condition in line~\ref{line:test} of Algorithm~\ref{algo:closure}, there exist $X \to A$ in ${\cal FD}$ and $x$ over $X$ such that  $\Delta_t \vdash xa$ and for every $b$ in $x$, $b \in cl^j$. Thus $\mu(x) \cap \mu(a) \ne \emptyset$ (because $\mu \models \Delta_t$) and $\mu(t) \subseteq \mu(b)$ for every $b$ in $x$ (by our induction hypothesis, because $b$ is in $cl^j$). Hence $\mu(t) \subseteq \mu(x)$ and $\mu(x) \subseteq \mu(a)$ hold, thus implying that $\mu(t) \subseteq \mu(a)$.

As a consequence, we have shown that for every $\mu$ such that $\mu \models \Delta$, for every $a$ in $cl(t)$, $\mu(t) \subseteq \mu(a)$. Therefore, by Definition~\ref{def:closure}, $cl(t) \subseteq t^+$ holds.

\smallskip
Conversely, $t^+ \subseteq cl(t)$ is shown by contraposition: assuming that $a \not\in cl(t)$, we prove that $a \not\in t^+$. To this end, we exhibit a ${\cal T}$-mapping $\mu_t$ such that $\mu_t \models \Delta$ and $\mu_t(t) \not\subseteq \mu_t(a)$.

We denote by $\mu_t^*$ the ${\cal T}$-mapping built up as $\mu^*$, but starting from $\Delta_t$ as defined line~\ref{line:cl-db-assign} in Algorithm~\ref{algo:closure}. Thus, $\mu_t^* \models \Delta_t$, and since $\mu^*_t(t) \ne \emptyset$, it has been seen above that $\mu^*_t \models \Delta$.

Thus, if $\mu^*_t(t) \not\subseteq \mu^*_t(a)$ then $\mu_t^*$ is the ${\cal T}$-mapping we are looking for, and thus, we set $\mu_t = \mu^*_t$. Assuming that $\mu^*(t) \subseteq \mu^*(a)$, let $k$ be an integer not in $\mu^*_t(\alpha)$ for any $\alpha$ occurring in $\Delta_t$, and let $\mu_t$ be the ${\cal T}$-mapping defined for every constant $\alpha$ by:
\\
$-$ $\mu_t(\alpha) = \mu^*_t(\alpha) \cup \{k\}$, if $\alpha \in cl(t)$
\\
$-$ $ \mu_t(\alpha) = \mu^*_t(\alpha)$, otherwise.
\\
We show that $\mu_t$ satisfies that: $(1)$ $\mu_t(t) \not\subseteq \mu_t(a)$ and $(2)$ $\mu_t \models \Delta$. 

\smallskip\noindent
$(1)$ Since every $\alpha$ in $t$ is in $cl(t)$, $k$ is in $\mu_t(t)$ and since $a$ is not in $cl(t)$, $k$ is not in $cl(a)$. It thus follows that $\mu_t(t) \not\subseteq \mu_t(a)$.

\smallskip\noindent
$(2)$ Since for every constant $\alpha$, $\mu^*_t(\alpha) \subseteq \mu_t(\alpha)$ holds, for every $q$ in $D$, it holds that $\mu^*_t(q) \subseteq \mu_t(q)$,  which implies $\mu_t(q) \ne\emptyset$, because $\mu^*_t(q) \ne \emptyset$ holds as a consequence of $\mu^*_t\models \Delta$.

To prove that $\mu_t \models Y \to B$ for every $Y \to B$ in ${\cal FD}$, let $y$ over $Y$ and $b$ in $dom(B)$ such $\mu_t(y) \cap \mu_t(b) \ne \emptyset$. To show that $\mu_t(y) \subseteq \mu_t(b)$, we consider the two cases according to which $\mu^*_t(y) \cap \mu^*_t(b)$ is or not empty.
\\
$\bullet$ If $\mu^*_t(y) \cap \mu^*_t(b) =\emptyset$, then by definition of $\mu_t$, for $\mu_t(y) \cap \mu_t(b)$ to be nonempty, it must be that $\mu_t(y) =\mu^*_t(y) \cup \{k\}$ and $\mu_t(b) =\mu^*_t(b) \cup \{k\}$. Writing $y$ as $\beta_1\ldots \beta_p$, this implies that every $\beta_i$ ($i =1, \ldots , p$), and $b$ are in $cl(t)$. Then, as we know that $cl(t) \subseteq t^+$ holds, all these constants are in $t^+$, implying that $\mu_t^*(t) \subseteq \mu_t^*(\beta_i)$ ($i=1, \ldots ,p$) and $\mu_t^*(t) \subseteq \mu_t^*(b)$, because $\mu_t^* \models \Delta$. Since $\mu_t^*(t) \ne \emptyset$, we have $\mu^*_t(y) \cap \mu^*_t(b) \ne \emptyset$, which contradicts our hypothesis that $\mu^*_t(y) \cap \mu^*_t(b) = \emptyset$. This case in thus not possible.
\\
$\bullet$ If $\mu^*_t(y) \cap \mu^*_t(b) \ne \emptyset$, then as $\mu^*_t \models {\cal FD}$, $\mu^*_t(y) \subseteq \mu^*_t(b)$ holds, and by Lemma~\ref{lemma:least-model} applied to $\Delta_t$, we also have that $\Delta_t \vdash yb$. Since $\mu^*_t(y) \subseteq \mu^*_t(b)$ holds, assuming that $\mu_t(y) \subseteq \mu_t(b)$ does not hold implies that $k$ belongs to $\mu_t(y)$ but not to $\mu_t(b)$. Hence, every $\beta_i$  ($i=1, \ldots ,p$) is in $cl(t)$ whereas $b$ is not. This is a contradiction with line~\ref{line:test} of Algorithm~\ref{algo:closure}, where it is stated that $\beta$ is inserted into $cl(t)$ (because $\Delta_t \vdash yb$ and every $\beta_i$  ($i=1, \ldots ,p$) is in $cl(t)$). Thus, $\mu_t(y) \subseteq \mu_t(b)$ holds showing that $\Delta_t \models Y \to B$. The proof is therefore complete.
\hfill$\Box$
\end{proof}
\section{Proof of Lemma~\ref{lemma:chase}}\label{append:lemma-chase}
{\bf Lemma~\ref{lemma:chase}.~}
{\em
Algorithm~\ref{algo:chase} applied to $\Delta=(D, {\cal FD})$ always terminates. Moreover, for every tuple $t$, $\mu^*(t)\ne \emptyset$ holds if and only if  $t$ is in ${\sf LoCl}(D^*)$.
}
\begin{proof}
The tuples inserted into $D^*$ when running the while-loop  line~\ref{line:main-loop-chase} of Algorithm~\ref{algo:chase} are built up using only constants occurring in $\Delta$. Thus, the number of these tuples is finite, and so, Algorithm~\ref{algo:chase} terminates.

The proof that for every $t$ in ${\sf LoCl}(D^*)$, $\mu^*(t) \ne\emptyset$ holds is conducted by induction on the steps of Algorithm~\ref{algo:chase}. If $(D_k)_{k \geq 0}$ denotes the sequence of the states of $D^*$ during the execution, we first note that since $D_0 = D$, for every $t$ in ${\sf LoCl}(D_{0})$,  $\mu^*(t) \ne \emptyset$ holds.

Assuming now that for $i>0$, for every $t$ in ${\sf LoCl}(D_{i})$, $\mu^*(t) \ne\emptyset$, we prove the result for every $t$ in ${\sf LoCl}(D_{i+1})$. Indeed, let $t'$ in $D_{i+1}$ such that $t \sqsubseteq t'$. If $t'$ is in $D_i$, the proof is immediate; we thus now assume that $t'$ is not in $D_i$, that is that $t'$ occurs in $D_{i+1}$ when running Algorithm~\ref{algo:chase}, that is, there exist $X \to A$ in ${\cal FD}$, $t_1$ and $t_2$ in $D_i$ such that $t_1.X=t_2.X=x$, $t_1.A=a$ and either $(i)$ $t_2.A$ is not defined or $(ii)$ $t_2.A$ is defined but not equal to $t_1.A$. Writing $t_1$ as $t'_1xa$, we have the following:

\smallskip\noindent
$(i)$ If $t_2.A$ is not defined, then $t_2$ is written as $t'_2x$ and, according to the statement line~\ref{line:add-plus}, $t'$ is of the form $t'_2xa$. By our induction hypothesis, $\mu^*(t_1)$ and $\mu^*(t_2)$ are nonempty, and thus $\mu^*(x) \cap \mu^*(a) \ne \emptyset$. Hence, $\mu^*(x) \subseteq \mu^*(a)$ (because $\mu^* \models X \to A$), and so, $\mu^*(t')= \mu^*(t'_2) \cap \mu^*(x) \cap \mu^*(a)= \mu^*(t'_2) \cap \mu^*(x)$, showing that $\mu^*(t')= \mu^*(t_2)$. Hence $\mu^*(t')\ne \emptyset$, and so, $\mu^*(t) \ne \emptyset$ also holds, since $\mu^*(t')\subseteq \mu^*(t)$.

\noindent
$(ii)$ If $t_2.A$ is defined but $t_1.A \ne t_2.A$. for $i=1,2$, $t_i$ is written as $t'_ixa_i$ where $a_i = t_i.A$.  statement line~\ref{line:add-plus-bis}, $t'$ is one of the tuples $t'_1xa_2$ or $t'_2xa_1$, and each of these cases can be treated as in $(i)$ above,

\smallskip
We therefore have shown that if $t$ is in ${\sf LoCl}(D^*)$ as computed by the main loop line~\ref{line:main-loop-chase} of Algorithm~\ref{algo:chase}, then $\mu^*(t)\ne \emptyset$. Since the last loop line~\ref{line:norm+} does not change this  set ${\sf LoCl}(D^*)$, this part of the proof is complete.

\smallskip
Conversely, we show that for every $t$, if $\mu^*(t)\ne \emptyset$ then $t$ is in ${\sf LoCl}(D^*)$. The proof is done by induction on the construction of $\mu^*$. 
By definition of $\mu_0$, it is clear that if $\mu_0(t) \ne \emptyset$ then $t$ is in ${\sf LoCl}(D)$ and thus in ${\sf LoCl}(D^*)$. Now, if we assume that for every $i >0$ and every $t$, if $\mu_i(t) \ne \emptyset$ then $t$ belongs to ${\sf LoCl}(D^*)$, we prove that this result holds for $\mu_{i+1}$.

Let $t$ be such that $\mu_i(t) =\emptyset$ and $\mu_{i+1}(t) \ne \emptyset$. For every $\alpha$, writing $\mu_{i+1}(\alpha)$ as $\mu_i(\alpha) \cup M(\alpha)$, where $M(\alpha)$ is the union of all $\mu_i(x)$ such that $x$ is a tuple over $X$, where $X \to A \in {\cal FD}$, $ \alpha \in dom(A)$, $\mu_i(x) \cap \mu_i(\alpha) \ne \emptyset$, and $\mu_i(x) \not\subseteq \mu_i(\alpha)$, we have the following:

\smallskip\noindent
\begin{tabular}{rlr}
$\mu_{i+1}(t)$&$= \bigcap_{\alpha\sqsubseteq t}\mu_{i+1}(\alpha)$&\\
&$= \bigcap_{\alpha\sqsubseteq t}\left(\mu_{i}(\alpha) \cup M(\alpha)\right)$&\qquad(1)\\ 
&$= \mu_i(t) \cup \left(\bigcup_{t=t_1t_2}\left(\mu_i(t_1)\cap \left(\bigcap_{\beta\sqsubseteq t_2}M(\beta)\right)\right)\right) \cup \left(\bigcap_{\alpha \sqsubseteq t}M(\alpha)\right)$&(2)\\
\end{tabular}

\smallskip\noindent
Equality (2) above is obtained from (1) by applying the distributivity of intersection over union with the convention that $t=t_1t_2$ refers to any split of $t$ into two tuples $t_1$ and $t_2$. Assuming $\mu_i(t) =\emptyset$ and $\mu_{i+1}(t) \ne \emptyset$ implies that in Equality (2) either the second or the last term of the union is nonempty. 
\\
$\bullet$ If $\bigcup_{t=t_1t_2}\left(\mu_i(t_1)\cap \left(\bigcap_{\beta\sqsubseteq t_2}M(\beta)\right)\right) \ne \emptyset$, there exist $t_1$ and $t_2$ such that  $t=t_1t_2$ and $\mu_i(t_1)\cap \left(\bigcap_{\beta\sqsubseteq t_2}M(\beta)\right)\ne \emptyset$. Given such a split of $t$, writing $t_2$ as $\beta_1\ldots \beta_p$ implies that, for $k=1, \ldots , p$, $M(\beta_k)$ contains $y_k$ such that $Y_k \to B_k$ is in ${\cal FD}$ and $\mu_i(y_k) \cap \mu_i(\beta_k) \ne \emptyset$. Moreover, we have that $\mu_i(t_1)\cap \left(\bigcap_{k=1}^{k=p}\mu_i(y_k)\right) \ne \emptyset$.
Thus  by our induction hypothesis, ${\sf LoCl}(D^*)$ contains a tuple of the form $q_1t_1y_1\ldots y_p$ and $p$ tuples of the form $q'_ky_k\beta_k$ ($k=1, \ldots ,p$).

Now, given $k=1, \ldots, p$, if $q_1t_1y_1\ldots y_p$ is not defined over $B_k$, $q_1t_1y_1\ldots y_p\beta_k$ appears in $D^*$ due to the statement line~\ref{line:add-plus} of Algorithm~\ref{algo:chase}. Assume now that $q_1t_1y_1\ldots y_p$ is defined over $B_k$ but with a value different than $\beta_k$, say $\beta'_k$.

By construction of $t_1$ and $t_2$, $B_k$ is not in $sch(t_1)$, and so, $B_k$ is either in $sch(q_1)$ or in $Y_i$ for some $i=1, \ldots ,p$. In any case, denoting $sch(q_1y_1\ldots y_p)$ by $Q$, we write $q_1t_1y_1\ldots y_p$ as $r^kt_1b'_k$ where $r^k = (q_1y_1\ldots y_p).(Q \setminus B_k)$. Considering that $r^kt_1b'_k$ and $q'_ky_k\beta_k$ have the same $Y_k$-value $y_k$, the statement line~\ref{line:add-plus-bis} of Algorithm~\ref{algo:chase} applies and $r^kt_1\beta_k$ is inserted in $D^*$. During the subsequent iterations, a similar argument shows that $D^*$ contains a tuple of the form $rt_1\beta_1\ldots \beta_p$, that is $rt_1t_2$ or $rt$. It thus follows that $t$ is in ${\sf LoCl}(D^*)$. 
\\
$\bullet$ If $\bigcap_{\alpha \sqsubseteq t}M(\alpha) \ne \emptyset$, the same reasoning as above applies considering that $t_1$ is empty and $t_2=t$. After the iterations, $D^*$ contains a tuple of the form $r\beta_1\ldots \beta_p$, that is $rt$. Thus,  in this case again, $t$ is in ${\sf LoCl}(D^*)$, and the proof is complete. \hfill$\Box$
\end{proof}
\section{Proof of Proposition~\ref{prop:closure}}\label{append:proof-prop-closure}
{\bf Proposition~\ref{prop:closure}.}~
Let $\Delta =(D, {\cal FD})$ and $t$ be such that $\Delta \vdash t$. For every tuple $q$ and every  $a$ in $dom(A)$ such that $q \sqsubseteq t$ and $a \sqsubseteq t$, we have: $a$ belongs to $q^+$ if and only if $A$ belongs to $Q^+$.
\begin{proof}
Assuming first $a$ in $q^+$, we show by induction on the steps of Algorithm~\ref{algo:closure} that $A$ is in $Q^+$. It is important to notice that since $q \sqsubseteq t$ and $\Delta \vdash t$, $\Delta \vdash q$ holds. Hence when running Algorithm~\ref{algo:closure} with $\Delta$ and $q$ as input, as shown in the proof of Lemma~\ref{lemma:inclusion}, $\mu \models \Delta$ holds if and only if $\mu \models \Delta_q$ holds. Thus, for every tuple $\tau$, $\Delta \vdash \tau$ holds if and only if $\Delta_q \vdash \tau$ holds.

\smallskip
If $a$ is in $q^+$ because of line~\ref{line:init-cl} in Algorithm~\ref{algo:closure}, then $A$ is in $Q$, showing that $A$ is in $Q^+$. If $a$ is inserted in $q^+$ because of line~\ref{line:test}, then there exist $X \to A$ in ${\cal FD}$ and $x$ over $X$ such that every $b$ in $x$ belongs to $q^+$ and $\Delta_q \vdash xa$, that is $\Delta \vdash xa$. Assuming that the proposition holds for every $b$ in $x$ implies that every $B$ in $X$ is in $Q^+$. Thus, $X \subseteq Q^+$ holds, and so $A$ is in $Q^+$.

Conversely, let $A$ be in $Q^+$. If $A$ is in $Q$, then $q.A=a$, and so, $a$ is in $q^+$. Let us now assume that $A$ is not in $Q$, and let us show by induction on the execution of the loop computing $Q^+$ that $a$ belongs to $q^+$. Indeed,  denoting by $Q'$ the current value of $Q^+$ when $A$ is inserted in $Q^+$, there exists $X \to A$ in ${\cal FD}$ such that  $X \subseteq Q'$. Thus, by our induction hypothesis,  every $\alpha$ in $q.X$ is in $q^+$. Moreover, since $\Delta \vdash t$ and $xa=t.XA$, $\Delta \vdash xa$. Hence, $\Delta_q \vdash xa$, and by the statement line~\ref{line:test} of Algorithm~\ref{algo:closure}, $a$ belongs to $q^+$. The proof is therefore complete.\hfill$\Box$
\end{proof}
\section{Proof of Lemma~\ref{lemma:incons}}\label{append:proof-prop-chase}
{\bf Lemma~\ref{lemma:incons}.}~
{\em
Given  $\Delta=(D, {\cal FD})$, a tuple $t$ is inconsistent in $\Delta$ if and only if $t \in  {\sf Inc}(\Delta)$.
}
\begin{proof}
We note first  that for every $x$ in $inc(X \to A)$ there exist $a_1, \ldots ,a_k$  ($k \geq 2$) in $dom(A)$ such that for every $i=1, \ldots , k$, $xa_i \in {\sf LoCl}(D^*)$, thus such that $\Delta \vdash xa_i$. Therefore, for every $i=1, \ldots , k$, $a_i$ belongs to $x^+$, and so, $\Delta \vdash (x \preceq a_1 \sqcap \ldots \sqcap a_k)$ holds, showing that $x$ is inconsistent in $\Delta$.

We now prove that if $q$ belongs to ${\sf Inc}(\Delta)$  then $q$ is inconsistent in $\Delta$. Indeed, by Algorithm~\ref{algo:incons}, there exist $t$ in $D^*$, $X \to A$ in ${\cal FD}$,  such that $Q \subseteq T$, $t.Q=q$, $t.X \in inc(X \to A)$, and $X \subseteq Q^+$. Since $\Delta \vdash t$, Proposition~\ref{prop:closure} applies, showing that for every $\alpha$ in $x$, $\alpha$ belongs to $q^+$, where $q=t.Q$. Hence,  every $a_i$ in $x^+$ is also in $q^+$, and thus for every $i=1, \ldots , k$, $\Delta \vdash (q \preceq a_i)$, implying that $q$ is inconsistent in $\Delta$.

Conversely, if $q$ is inconsistent in $\Delta$, then $\Delta \vdash q$ and $\Delta \dashapprox q$. Thus, there exist $A$ in $U$ and $a$ and $a'$ in $dom(A)$ such that $\Delta \vdash (q \preceq a \sqcap a')$, implying that $\Delta \vdash qa$ and $\Delta \vdash qa'$. By Lemma~\ref{lemma:chase}, $D^*$ contains two rows $t$ and $t'$ such that $qa \sqsubseteq t$ and $qa' \sqsubseteq t'$. This implies that $A$ can not be in $Q$ because otherwise, we would for instance have $qa=q$ and thus $qa'=qaa'$, which does not define a tuple. Since,  by Definition~\ref{def:closure}, $\Delta \vdash (q \preceq a \sqcap a')$ implies  that $a$ and $a'$ are in $q^+$, by Proposition~\ref{prop:closure}, $A$ is in $Q^+$. Since $A$ is not in $Q$, ${\cal FD}$ contains $X \to A$ such that $X \subseteq Q^+$. It follows that $A$ is in $X^+$, $t.XA=xa$ and $t'.XA=xa'$. Therefore $x$ belongs to $inc(X \to A)$.

Summing up, we have found a tuple $t$ in $D^*$ and $X \to A$ in ${\cal FD}$ such that $t.X$ belongs to $inc(X \to A)$, $q \sqsubseteq t$ and $X \subseteq Q^+$. It thus follows from line~\ref{line:ins-inc} of Algorithm~\ref{algo:incons} that $q$ belongs to ${\sf Inc}(\Delta)$, which completes the proof.\hfill$\Box$
\end{proof}
\section{Proof of Proposition~\ref{prop:cons-answers}}\label{append:proof-cons-answers}
{\bf Proposition~\ref{prop:cons-answers}.}~
{\em
Given $\Delta=(D, {\cal FD})$ and a query {\rm $Q: {\sf SELECT}$ $X$ {\sf [WHERE $\Gamma$]}}, Algorithm~\ref{algo:answer-repair} correctly computes $ans^\downarrow_\Delta(Q)$ and $ans^\uparrow_\Delta(Q)$.
Moreover, the following holds:
$ans^\downarrow_\Delta(Q) \subseteq ans^\uparrow_\Delta(Q) \subseteq ans^+_\Delta(Q)$.
}
\begin{proof}
In this proof we respectively denote by $ans^\downarrow$ and $ans^\uparrow$ the two sets returned by Algorithm~\ref{algo:answer-repair} and we successively show that $ans^\downarrow =ans^\downarrow (Q)$ and $ans^\uparrow =ans^\uparrow (Q)$.

First it is clear that all selected tuples are defined over $X$ and that they satisfy $\Gamma$. Moreover, assuming that the previous two conditions are satisfied, a tuple $t$ generates an $X$-value in $ans^\downarrow_\Delta(Q)$, if and only if $t$ is in every repair $R$ of $\Delta$, that is if and only if $t$ contains no conflicting value with respect to some dependency in ${\cal FD}$. This condition being precisely that on line~\ref{line:test1} of Algorithm~\ref{algo:answer-repair}, we obtain that $ans^\downarrow = ans^\downarrow_\Delta(Q)$. 

\smallskip
Now, given $x$ in  $ans^\uparrow_\Delta(Q)$,  assume that the condition on line~\ref{line:test2} is not satisfied. In this case there exist $t$ in $D^*$ and $Y \to B$  in ${\cal FD}$ such that $t$ satisfies $\Gamma$, $t.Y$ is in $inc(Y \to B)$, $t.X=x$ and $t.B$ occurs in $x$. Thus, by the statement on line~\ref{line:add-plus} in Algorithm~\ref{algo:chase}, $D^*$ contains a tuple $t'$ such that $sch(t) \subseteq sch(t')$, $t'.Y=t.Y$ and $t'.B \ne t.B$. Hence, writing $x$ as $x'b$, $Rep(\Delta)$ contains a repair $R$ where $x'b$ occurs and a repair $R'$ where $x'b$ does not occur, showing that $x$ cannot belong to $ans^\uparrow_\Delta(Q)$. This is a contradiction showing that $ans^\uparrow_\Delta(Q)\subseteq ans^\uparrow$ holds.

Conversely, we first notice that for every tuple $q$ occurring in a repair $R$ but not in another repair $R'$, there exist $q'$ in $R'$, $B$ in $sch(q)$ and $Y \to B$ in ${\cal FD}$ such that $q.Y=q'.Y=y$, $y \in inc(Y \to B)$ and $q.B\ne q'.B$. Now, if $x$ is a tuple over $X$ for which the condition on line~\ref{line:test2} is satisfied, then there exists $t$ in $D^*$ such that $t$ satisfies $\Gamma$, $t.X=x$ and for every $Y \to B$ in ${\cal FD}$ such that $y$ is in $inc(Y \to B)$, $B$ is not in $X$. Therefore, it turns out that $x=t.X$ occurs in $\pi_X(\sigma_\Gamma(R))$ for every $R$ in $Rep(\Delta)$, which shows that $x$ is in $ans^\uparrow_\Delta(Q)$.

\smallskip
As for the inclusions, in Algorithm~\ref{algo:answer-repair}, the condition on line~\ref{line:test1} implies that on line~\ref{line:test2}, showing the first inclusion. Moreover, this second condition implies the one on line~\ref{line:test0} of Algorithm~\ref{algo:answer}, showing the second inclusion. The proof is therefore complete.\hfill$\Box$
\end{proof}

\end{document}